\documentclass[12pt]{iopart}
\expandafter\let\csname equation*\endcsname\relax
\expandafter\let\csname endequation*\endcsname\relax
\usepackage{amssymb}
\usepackage{amsmath}
\usepackage{amsthm}
\usepackage[numbers,sort&compress,square]{natbib}
\usepackage{bm}
\usepackage{graphicx}
\usepackage{calc}
\usepackage[dvipsnames]{xcolor}
\usepackage{hyperref}
\usepackage[capitalise]{cleveref}
\usepackage{tabularray}

\newtheorem{theorem}{Theorem}[section]
\newtheorem{lemma}[theorem]{Lemma}
\newtheorem{proposition}[theorem]{Proposition}

\newtheorem{claim}{Claim}

\theoremstyle{definition}
\newtheorem{definition}[theorem]{Definition}

\newcommand{\Gin}{G_\text{in}}
\newcommand{\Gout}{G_\text{out}}
\newcommand{\Vin}{V_\text{in}}
\newcommand{\Vout}{V_\text{out}}
\newcommand{\Ein}{E_\text{in}}
\newcommand{\Eout}{E_\text{out}}
\newcommand{\EoutTilde}{\tilde{E}_\text{out}}
\newcommand{\nin}{n_\text{in}}
\newcommand{\nout}{n_\text{out}}
\newcommand{\noutTilde}{\bar{n}_\text{out}}
\newcommand{\Ain}{\mathbf{A}_\text{in}}
\newcommand{\Aout}{\mathbf{A}_\text{out}}
\newcommand{\AoutTilde}{\tilde{\mathbf{A}}_\text{out}}
\newcommand{\ssin}{s_\text{in}}
\newcommand{\ssout}{s_\text{out}}

\newcommand{\xin}{\mathbf{x}_\text{in}}
\newcommand{\xout}{\mathbf{x}_\text{out}}
\newcommand{\xoutTilde}{\tilde{\mathbf{x}}_\text{out}}
\newcommand{\ccin}{c_\text{in}}
\newcommand{\ccout}{c_\text{out}}
\newcommand{\Sk}{\mathrm{Sk}}

\begin{document}

\title[Integrability of Graph Embeddings]{Integrability of Generalised Skew-Symmetric Replicator Equations via Graph Embeddings}

\author{Matthew Visomirski$^1$ and Christopher Griffin$^2$}
\address{
	$^1$Department of Physics,
	The Pennsylvania State University,
    University Park, PA 16802
    }
\address{
	$^2$Applied Research Laboratory,
	The Pennsylvania State University,
    University Park, PA 16802
    }

\eads{\mailto{mav5583@psu.edu}, \mailto{griffinch@psu.edu}}
\date{\today~-~Preprint}

\begin{abstract} It is known that there is a one-to-one mapping between oriented directed graphs and zero-sum replicator dynamics (Lotka-Volterra equations) and that furthermore these dynamics are Hamiltonian in an appropriately defined nonlinear Poisson bracket.  In this paper, we investigate the problem of determining whether these dynamics are Liouville-Arnold integrable, building on prior work graph in graph decloning  by Evripidou  et al. [J. Phys. A., 55:325201, 2022] and graph embedding by Paik and Griffin [Phys. Rev. E. 107(5): L052202, 2024]. Using the embedding procedure from Paik and Griffin, we show (with certain caveats) that when a graph producing integrable dynamics is embedded in another graph producing integrable dynamics, the resulting graph structure also produces integrable dynamics. We also construct a new family of graph structures that produces integrable dynamics that does not arise either from embeddings or decloning. We use these results, along with numerical methods, to classify the dynamics generated by almost all oriented directed graphs on six vertices, with three hold-out graphs that generate integrable dynamics and are not part of a natural taxonomy arising from known families and graph operations. These hold-out graphs suggest more structure is available to be found. Moreover, the work suggests that oriented directed graphs leading to integrable dynamics may be classifiable in an analogous way to the classification of finite simple groups, creating the possibility that there is a deep connection between integrable dynamics and combinatorial structures in graphs. 
\end{abstract}

\vspace{2pc}

\submitto{\jpa}
\maketitle

\section{Introduction}

Conserved quantities are fundamental in both classical and modern physics with the correspondence between continuous symmetries and conserved quantities given by Noether's theorem \cite{A13}, providing a theoretical foundation for many of our theories \cite{W64,B06}. Dynamical systems exhibiting a sufficient number of conserved quantities (and generated by an appropriate Poisson bracket) are Liouville-Arnold integrable, admitting solutions that trace out foliating tori in phase space \cite{L-GPV12}. Integrable systems are rare, in the sense that a randomly chosen system of differential equations is unlikely to be Hamiltonian, let alone integrable. In light of Smale's comments and analysis \cite{S76}, it is surprising to find large families of integrable dynamics appearing in evolutionary dynamics, yet it is known that there are classes of Lotka-Volterra dynamics that are integrable in the Liouville-Arnold sense (see \cite{I87,I08,BIY08,EKV21,EKV22,PG23} for examples).

Let $\mathbf{A} \in \mathbb{R}^{n \times n}$ be a payoff (interaction) matrix. The replicator dynamics \cite{HS98,HS03} are the system of differential equations with form,
\begin{equation*}
    \dot{x}_i = x_i\left(\mathbf{e}_i^T\mathbf{A}\mathbf{x} - \mathbf{x}^T\mathbf{A}\mathbf{x}\right),
\end{equation*}
where $\mathbf{x} = \langle{x_1,\dots,x_n}\rangle$ is a vector of species proportions and $\mathbf{e}_i$ is the $i^\text{th}$ standard basis vector. As such $\mathbf{x} \in \Delta_{n-1}$, where,
\begin{equation*}
    \Delta_{n-1} = \left\{\mathbf{x} \in \mathbb{R}^n : \sum_i x_i = 1, x_i \geq 0\right\},
\end{equation*}
is the $n-1$ dimensional unit simplex embedded in $\mathbb{R}^n$. In a game-theoretic interpretation, $\mathbf{A}$ is a payoff (reward) matrix resulting from the interaction of two species; i.e., playing a game. See Hofbauer and Sigmund or \cite{HS98, HS03}, Weibull \cite{W97} or Tanimoto \cite{T15,T19} for a complete introduction to evolutionary games, and the replicator dynamics.

It is well known that the replicator dynamics are diffeomorphic to the generalised Lotka-Volterra equations \cite{HS98,HS03}. In the special case that $\mathbf{A}$ is skew-symmetric, the replicator equations simplify to,
\begin{equation}
    \dot{x}_i = x_i\left(\mathbf{e}_i^T\mathbf{A}\mathbf{x}\right),
    \label{eqn:Replicator}
\end{equation}
and in this case, we have an instance of the Lotka-Volterra equations with no need for a specialised diffeomorphism \cite{EKV22}. This class of evolutionary games models two-player, zero-sum games like rock-paper-scissors and was originally studied by Akin and Losert \cite{AL84}. Interestingly, this class of model is used extensively in theoretical ecology \cite{AL11,GBMA17,MCLM23} when the entries of $\mathbf{A}$ are restricted to $0, \pm 1$. Surprisingly, these models also find use in theoretical physics, with these dynamics occurring in the analysis of Schr\"{o}dinger operator \cite{VS93} in the work by Veselov and Shabat and in the discrete Korteweg-De Vries (KdV) equation (the Volterra lattice) as analysed by Moser \cite{M74} and Kac and Moerbeke \cite{KM75}, with generalisations by Bogoyavlenskij \cite{B88}.

Any skew-symmetric matrix $\mathbf{A}$ can be represented as an oriented directed graph $G_\mathbf{A} = (V,E)$, where $V = \{1,\dots,n\}$ and there is a directed edge $(i,j) \in E \subset V \times V$ precisely when $\mathbf{A}_{ij} < 0$ with corresponding weight $A_{ji}$. From a high-level ecological perspective, this interpretation is most sensible when we restrict the entries of $\mathbf{A}$ to $0, \pm 1$ and in this case, we have $A_{ij} = -1$ precisely when species $j$ preys on species $i$ as represented by the directed edge $(i,j)$ in $G_\mathbf{A}$. As a consequence of the mappings between the graph $G_\mathbf{A}$, the skew-symmetric matrix $\mathbf{A}$ and the replicator (Lotka-Volterra) dynamics given in \cref{eqn:Replicator}, we say that a graph $G_\mathbf{A}$ is (Liouville-Arnold) integrable if and only if the corresponding (Liouville-Arnold) dynamical system is integrable.

In this paper, we restrict our attention to the (partial) integrability of the replicator (Lotka-Volterra) dynamics arising from the case when the matrix $\mathbf{A}$ is both skew-symmetric and has entries $0$ or  $\pm 1$. In this work, we build on and use the prior work of Itoh \cite{I87,I08}, Bogoyavlenskij , Itoh and Yukawa \cite{BIY08} and Evripidou , Kassotakis and Vanhaecke \cite{EKV21,EKV22} and Griffin and Paik \cite{PG23} who along with others, have done extensive work on this class of systems \cite{CDPV15,DEKV17}. The main results of this paper, which we make precise in the sequel, are as follows.
\begin{enumerate}
\item  We characterise a new infinite family of integrable graphs that we call the \textit{skip vertex graphs}. This extends the work of Bogoyavlenskij  \cite{B88,BIY08} and hints at a potentially more general theory of integrability for these systems.

\item We show under certain conditions, if an inner graph $G_\text{in}$ is embedded in an outer graph $G_\text{out}$ and both graphs are integrable, then the resulting graph is integrable. We generalise this result to multiple simultaneous embeddings, defining the concept of an embedding graph.

\item Using these results, previous results on graph morphisms developed by  Evripidou  et al. \cite{EKV22}, as well as numerical methods, we organise the behaviour of all replicator (Lotka-Volterra) dynamics generated by directed graphs up to six vertices. Interestingly, there are still three ``hold-out'' graphs that produce integrable dynamics, but are not members of the integrable families we discuss, suggesting a more complete characterisation is needed.  
\end{enumerate}
This paper represents another step toward the ultimate goal of characterising the (zero-sum) replicator dynamics generated from directed graphs as a function of the graph structures themselves, largely begun with the early work on the Volterra lattice \cite{M74,KM75}. Our ultimate goal is to characterise the combinatorial \cite{I87,I08} structures that lead to integrable (or partially integrable) dynamics, and thus to provide a new and potentially deep connection between combinatorics and dynamics. 

The remainder of this paper is organised as follows: In \cref{sec:Prelim}, we provide preliminary definitions and results from prior work (e.g., \cite{I87,B88,EKV22,PG23}) that will be used throughout the rest of the paper. A new class of integrable graphs is presented in and the new class of integrable graphs in \cref{sec:SkipGraphs}. We present our main results on graph embeddings in \cref{sec:Embeddings} and discuss multiple simultaneous embeddings in \cref{sec:MultipleEmbeddings}. Using these results and numerical analysis, we provide the taxonomy of the dynamics arising from graphs containing up to six vertices in \cref{sec:Taxonomy}. Conclusions and future directions of research are presented in \cref{sec:Conclusion}. Mathematica notebooks containing code that assisted us in developing these results are provided as supplemental information (SI).

\section{Preliminary Definitions and Results} \label{sec:Prelim}
We provide mathematical preliminaries, most of which can be found in Evripidou  et al. \cite{EKV22} and Paik and Griffin \cite{PG23}. Material on Poisson structures is provided in detail in Laurent-Gengoux, Pichereau and Vanhaecke's book on the subject \cite{L-GPV12}. For brevity, we omit most details and refer the interested reader to the appropriate references. 

\subsection{Graph Operations}

We formalise the one-to-one relationship between directed graph structures and skew-symmetric matrices with entries restricted to $0,\pm1$. For this, recall that a directed graph $G = (V,E)$ is oriented the edge $(i,j)$ either does not appear in $E$ or if it does appear, then the edge $(j,i)$ does not appear in $E$. The following definition is adapted from Evripidou  et al. \cite{EKV22}. 

\begin{definition}[Skew Symmetric Graph] Let $\mathbf{A} \in \mathbb{R}^{n\times n}$ be a skew-symmetric matrix with entries consisting of only $0$ and $\pm 1$. The oriented directed graph $\mathbf{G}_\mathbf{A} = (V,E)$ has vertex set $V = \{1,\dots,n\}$ and edge set $E \subset V \times V$ so that,
\begin{equation*}
    (i, j) \in E \iff A_{ij} = -1.
\end{equation*}
\end{definition}
The following relationship is easy to see. 
\begin{proposition} There is a one-to-one correspondence between oriented directed graphs and skew-symmetric matrices with entries $0$,$\pm 1$. \hfill\qed
\end{proposition}
Evripidou  et al. \cite{EKV22} refer to this class of graphs as  skew-symmetric graphs. For the remainder of this paper, all graphs will be skew-symmetric and will just be called graphs. 

The ``embedding'' terminology used in the next definition is borrowed from the mathematical neuroscience community \cite{PMMC22,PALC+22}. It is somewhat unfortunate, since graph embeddings are a well-known part of topological graph theory \cite{GY18}. To reduce some confusion, we modify the terminology used in \cite{PMMC22,PALC+22} and \cite{PG23}.

\begin{definition}[Graph-to-Vertex Embedding] Let $\Gin = (\Vin,\Ein)$ and $G_\text{out}=(\Vout,\Eout)$ two graphs with disjoint vertex sets and suppose $v$ is a vertex in $\Gout$. Let $\EoutTilde \subset \Eout$ be the set of edges in $\Eout$ not containing $v$. The graph-to-vertex embedding graph $J = G_\text{in}\hookrightarrow_v G_\text{out}$ has vertex set $V = \Vin \sqcup (\Vout\setminus\{v\})$ and edge set $E = \Ein \sqcup \EoutTilde \sqcup E_\text{embed}$, where,
\begin{multline*}
    (i,j) \in E_\text{embed} \iff \\\left(i \in V_\text{in} \wedge j \in V_\text{out} \wedge (v,j) \in E_\text{out}\right) \vee \left(i \in V_\text{out} \wedge j \in V_\text{in} \wedge (i,v) \in E_\text{out}\right).
\end{multline*}
\label{def:Embed}
\end{definition}
Throughout the remainder of this paper, we will refer to this process simply as embedding, as there is now no chance of confusing it with topological graph embedding.

The process of (graph-to-vertex) embedding is illustrated in \cref{fig:GraphEmbedding} and was suggested by Paik (in \cite{PG23}) and inspired by Curto et al.'s cyclic union \cite{PMMC22,PALC+22}. Notice we are embedding into vertex $1$ of $\Gout$. In the resulting graph, we decorate the vertices remaining from $\Gout$ with tildes.
\begin{figure}[htbp]
\centering
\includegraphics[width=0.65\textwidth]{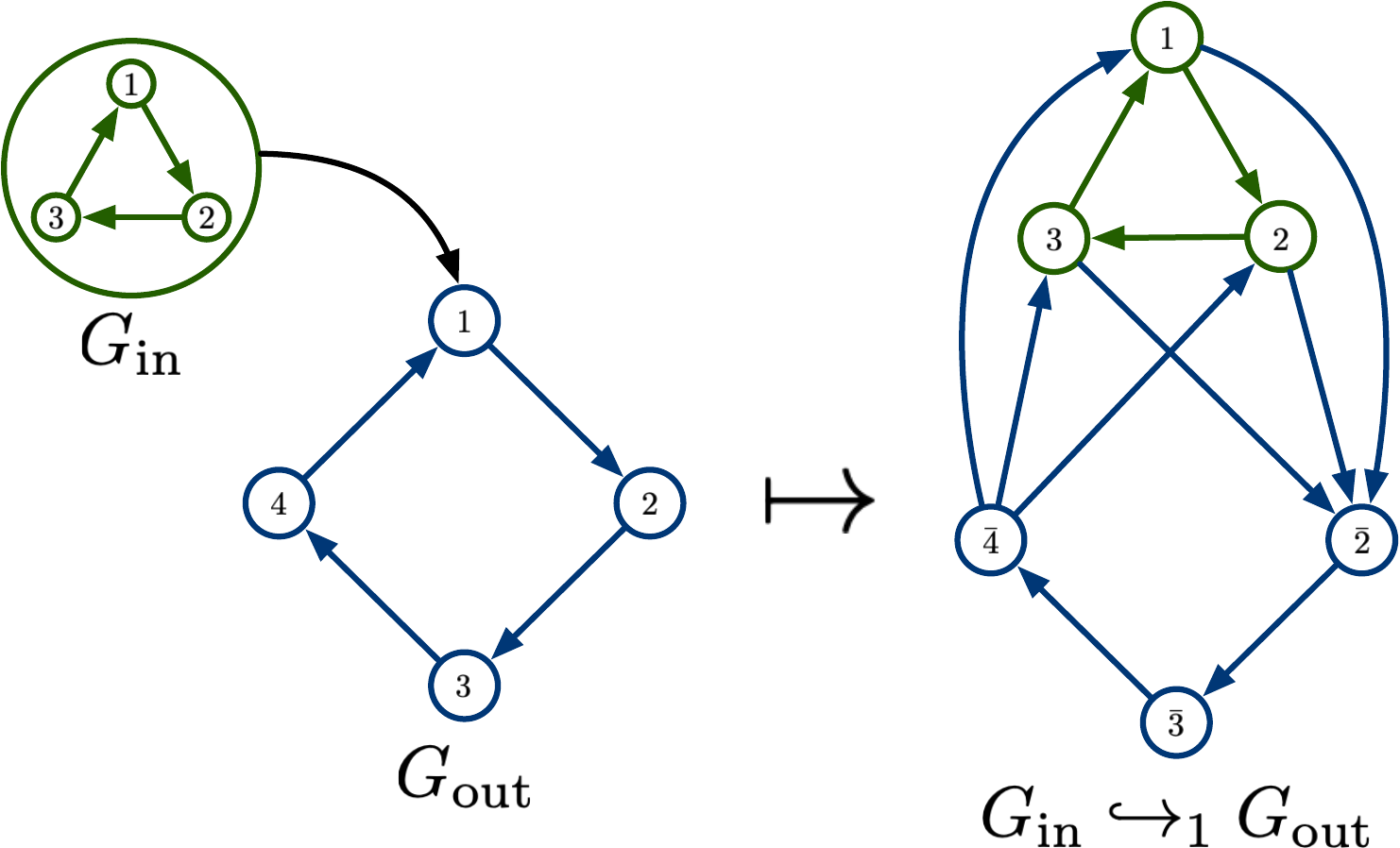}
\caption{The directed three-cycle is embedded into a directed four-cycle. The edges in the three cycle are preserved. The edges in the four-cycle to and from the replaced vertex are replicated to each vertex in the embedded three-cycle. The remaining vertices from the four-cycle are adorned with bars.}
\label{fig:GraphEmbedding}
\end{figure}
Let $\Ain \in \mathbb{R}^{\nin \times \nin}$ be the skew-symmetric matrix corresponding to $\Gin$ and let $\Aout \in \mathbb{R}^{\nout \times \nout}$ be the skew-symmetric matrix corresponding to $\Gout$. Let,
\begin{equation}
    \Aout = \begin{bmatrix} 0 & \mathbf{r}^T\\-\mathbf{r} & \AoutTilde\end{bmatrix},
    \label{eqn:Aout}
\end{equation}
where $\mathbf{r}^T$ is an appropriately sized row vector. Suppose we embed $\Gin$ into the vertex corresponding to the first row (column) of $\Aout$. That is, let $v = 1$ in \cref{def:Embed}. Then the skew-symmetric matrix for the graph embedding is given by,
\begin{equation}
    \mathbf{A} = \begin{bmatrix} \Ain & \mathbf{R}^T\\-\mathbf{R} & \AoutTilde
    \end{bmatrix},
\label{eqn:Aembed}
\end{equation}
with,
\begin{equation}
    \mathbf{R} = \begin{bmatrix} \mathbf{r} & \mathbf{r} & \cdots & \mathbf{r}\end{bmatrix},
\label{eqn:R}
\end{equation}
having $\nin$ columns all equal to $\mathbf{r}$. We will use this construction in the proof of \cref{thm:Embed}.

The following operation is defined by Evripidou  et al. \cite{EKV22} and will be used in our classification of dynamics arising from graphs with at most six vertices.   
\begin{definition}[Cloned Vertices] Suppose $G_\mathbf{A} = (V,E)$ is a graph corresponding to the matrix $\mathbf{A}$. A vertex $j \in V$ is a clone of $i \in V$ if $j$ and $i$ share the same neighbourhood in $G$. Equivalently, if row $i$ of $A$ and row $j$ of $A$ are identical, then $i$ and $j$ are clones.
\end{definition}

\begin{definition}[Irreducible Graph] A graph $G_\mathbf{A}$ is irreducible if $G_\mathbf{A}$ has no cloned vertices or equivalently if each row of $\mathbf{A}$ is unique.
\end{definition}

If a graph $G_\mathbf{A} = (V,E)$ has cloned vertices, then it is straightforward to see we can partition $V$ into subsets $V_1,\dots,V_n$ for some $n \geq 1$, where if $i,j \in V_k$, then $i$ and $j$ are clones. Since the vertices in $V_k$ share neighbourhoods, we can map $G_\mathbf{A}$ to a new irreducible graph $\tilde{G}_{\tilde{A}} = (\tilde{V},\tilde{E})$ with $\tilde{V} = \{1,\dots,n\}$ and $(i,j) \in \tilde{E}$ if and only if there is some $v \in V_i$ and $w \in V_j$ so that $(v,w) \in E$. It is straightforward to see that the new skew-symmetric matrix $\tilde{\mathbf{A}}$ is simply the matrix $\mathbf{A}$ with duplicate rows and their corresponding columns removed.

\begin{definition}[Decloning] Let $G_\mathbf{A}$ be a graph. The mapping $\eta : G_\mathbf{A} \mapsto \tilde{G}_A$ is the decloning operation.
\end{definition}

\cref{fig:Decloning} shows a 5 vertex graph on the left in which vertices $3$, $4$ and $5$ are clones. The decloned form of the graph is shown on the right.
\begin{figure}[htbp]
\centering
\includegraphics[width=0.65\textwidth]{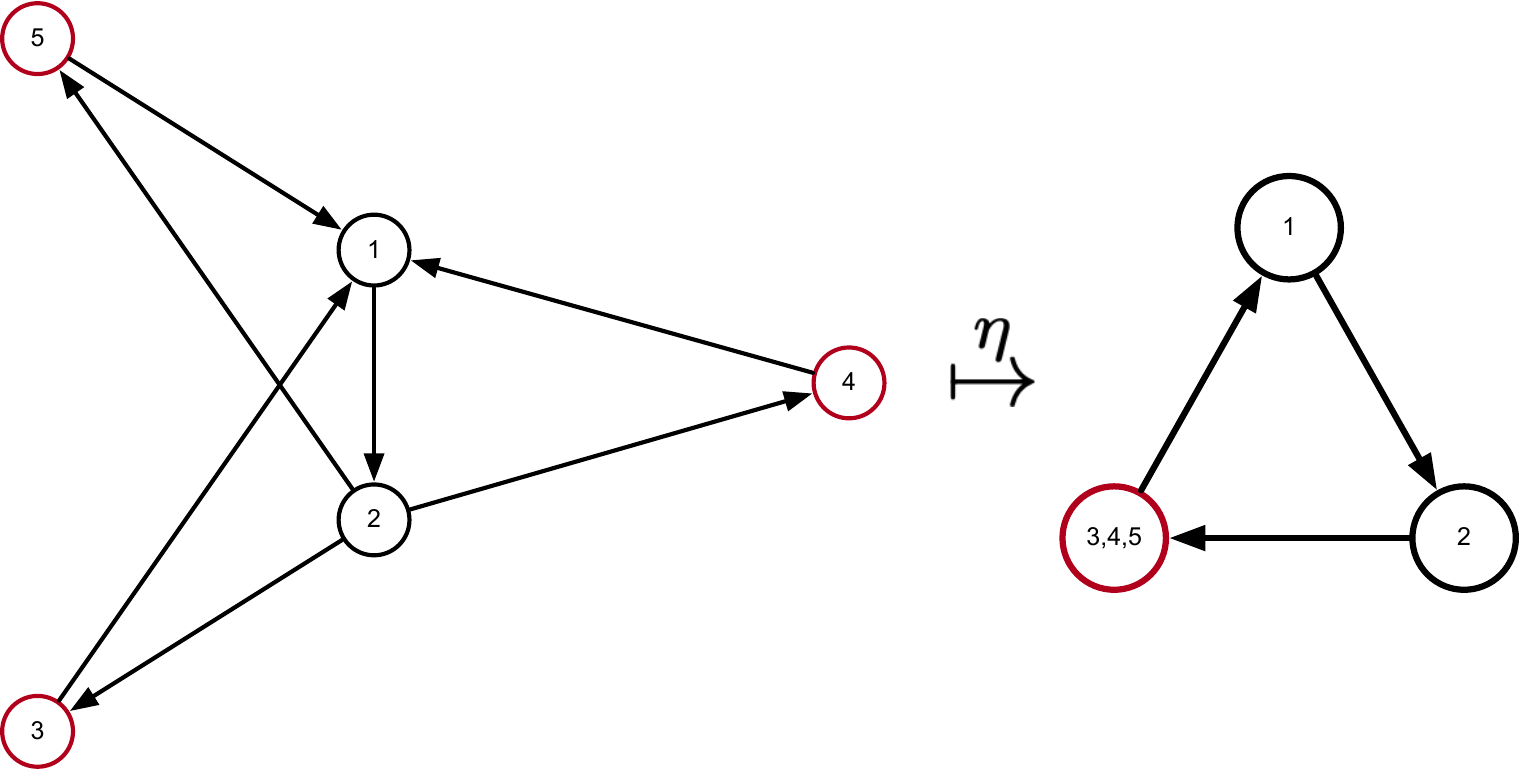}
\caption{An example of decloning a graph.}
\label{fig:Decloning}
\end{figure}
The decloning construction will be used as part of our classification scheme for the dynamics generated by graphs with fewer than 7 vertices.

We will frequently refer to a special class of graphs, sometimes referred to as Bogoyavlenskij -Itoh graphs, or sometimes just the Bogoyavlenskij  graphs \cite{EKV22}. 
\begin{definition}[Bogoyavlenskij  Graphs] The \textit{Bogoyavlenskij  graph} $B(n,k)$ where $k < \frac{n}{2}$ has vertex set $V = \{1,\dots,n\}$ and edge $E$ with edge $(i,j) \in E$ if $j = 1 + (i \oplus_n l)$ for $0 \leq l < k$, where $\oplus_n$ denotes addition modulo $n$.
\end{definition}
Put more simply, if we arrange the numbers in $\{1,\dots,n\}$ in a circle, then the graph $B(n,k)$ has edges from vertex $i$ to the next $k$ vertices working around the circle. This is illustrated in \cref{fig:B72}.
\begin{figure}[htbp]
\centering
\includegraphics[width=0.35\textwidth]{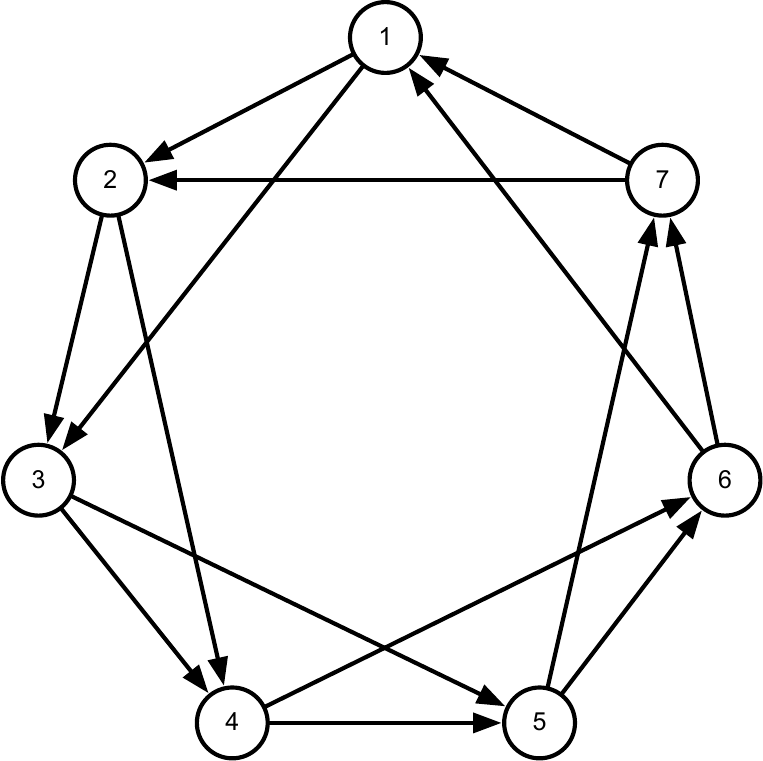}
\caption{The Bogoyavlenskij  graph $B(7,2)$ illustrates the definition of these graphs using geometric arrangement.}
\label{fig:B72}
\end{figure}

\subsection{Integrability of Replicator Systems}
Let $\mathbf{A} \in \mathbb{R}^{n \times n}$ be a skew-symmetric matrix and let $F,G:\mathbb{R}^n \to \mathbb{R}$ be differentiable functions. Define the bracket,
\begin{equation}
\{F,G\}_\mathbf{A} = \sum_{i < j} A_{ij}x_i x_j \left(\frac{\partial F}{\partial x_i} \frac{\partial G}{\partial x_j} - \frac{\partial G}{\partial x_i}\frac{\partial F}{\partial x_j}\right).
\label{eqn:Bracket}
\end{equation}
Direct computation shows that \cref{eqn:Replicator} is a Hamiltonian system with Hamiltonian $H_1(\mathbf{x}) = x_1 + \cdots + x_n$. That is, \cref{eqn:Replicator} can be rewritten as,
\begin{equation*}
    \dot{x}_i = \{x_i, H_1\}_\mathbf{A}.
\end{equation*}
and $H_1$ is a conserved quantity, i.e., $\dot{H}_1 = 0$, which is obvious for the replicator as we automatically have $H_1 = 1$, assuming we begin with an initial condition on $\Delta_{n-1}$. It is known that certain instances of the replicator (Lotka-Volterra) dynamics admit additional conserved quantities \cite{I87,I08,PG23}, which we will discuss in the sequel. As an interesting historical note, the nonlinear (quadratic) bracket defined in \cref{eqn:Bracket} is not well known in the broader literature and has been used and/or ``rediscovered'' on several occasions including (for example) by Itoh \cite{I87}, Griffin \cite{G21}, Ovsienko, Schwartz and Tabachnikov \cite{OST13}, even though it does appear in textbooks by Marsden and Ratiu \cite{MR13} and Laurent-Gengoux, Pichereau and Vanhaecke \cite{L-GPV12}.

The Poisson structure corresponding to this bracket,
\begin{equation*}
    \pi_\mathbf{A} = \sum_{i < j} A_{ij}x_i x_j\frac{\partial}{\partial x_i} \wedge \frac{\partial}{\partial x_j},
\end{equation*}
has been studied extensively by Evripidou  et al. \cite{EKV22} and is discussed in \cite{L-GPV12}. In particular, the rank of $\pi_\mathbf{A}$ precisely corresponds to the rank of $\mathbf{A}$. The following definition follows immediately from this fact and Def. 12.9 of \cite{L-GPV12}.

\begin{definition} Let $\mathbf{A} \in \mathbb{R}^{n \times n}$ be a skew-symmetric matrix with rank denoted $\mathrm{rank}(\mathbf{A})$. The replicator (Lotka-Volterra) dynamics, \cref{eqn:Replicator}, generated from $\mathbf{A}$ are \textit{integrable} if there are $s$ conserved quantities $H_1,\dots,H_s$ so that:
\begin{enumerate}
    \item $H_1,\dots,H_s$ are (algebraically) independent.
    \item $H_1,\dots,H_s$ are in involution, or commute under the bracket. That is, $\{H_i,H_j\}_\mathbf{A} = 0$ for all $i,j$. 
    \item The following relation among matrix rank (Poisson manifold rank), embedding dimension and number of conserved quantities holds,
    \begin{equation*}
        s + \tfrac{1}{2}\mathrm{Rank}(\mathbf{A}) = n.
    \end{equation*}
\end{enumerate}
\label{def:LAIntegrability}
\end{definition}
As noted in \cite{EKV22}, the Casimirs of the Poisson algebra corresponding to $\pi_\mathbf{A}$ can be read from the eigenvectors of $\mathbf{A}$ corresponding to the zero eigenvalue. To see this, note that if $\bm{\alpha} = \langle{\alpha_1,\dots,\alpha_n}\rangle$ is an eigenvector of $\mathbf{A}$ with eigenvalue $0$, then,
\begin{equation}
    \frac{d}{dt}\left(x_1^{\alpha_1} \cdots x_n^{\alpha_n}\right) = x_1^{\alpha_1} \cdots x_n^{\alpha_n}\left(\bm{\alpha}^T\mathbf{A}\mathbf{x}\right) = 0,
\label{eqn:TimeDerivativeOfProduct}
\end{equation}
and thus $x_1^{\alpha_1} \cdots x_n^{\alpha_n}$ is a conserved quantity of the system. It is worth noting that Akin and Losert's \cite{AL84} noted that these quantities were conserved in the context of zero-sum replicator dynamics, but studied them purely from the perspective of symplectic geometry in which they were treated as the Hamiltonians. In unrelated work, Hamiltonian structure in evolutionary games was also considered by Hofbauer \cite{H96}. Neither Hofbauer nor Akin and Losert consider general integrability. 

In general, there will be more conserved quantities than just the Casimirs, as shown in special cases by Kac and Moser \cite{KM75} and Moerbeke \cite{M74}, Bogoyavlenskij  \cite{B91,BIY08} and especially by Itoh \cite{I87,I08}, who provided a combinatorial analysis to the construction of conserved quantities in Bogoyavlenskij  graphs. In particular, the following theorem is immediate from the work of Itoh and  Bogoyavlenskij  \cite{I87,I08,BIY08}.
\begin{theorem} Let $G = B(n,k)$ be a Bogoyavlenskij  graph. Then $G$ admits a sufficient number of conserved quantities for integrability. \hfill\qed
\end{theorem}
We note that the integrability of the Bogoyavlenskij  graphs is only known for directed cycles (via the work Moser \cite{M74}, Kac and van Moerbeke \cite{KM75}) and balanced tournaments \cite{VS93,I87,B91}. However, we (and others) conjecture that these graphs are all integrable. In particular, direct computation shows that if $n \leq 6$, then $B(n,k)$ is integrable for all appropriate $k$ (this is included in a Mathematica file in the SI). We will use this fact in our taxonomy of graphs on six or fewer vertices in \cref{sec:Taxonomy}.

Itoh's combinatorial description \cite{I08} of the conserved quantities can be rephrased using graph theoretic language. We do so for the class $B(n,1)$, which is equivalent to the class of directed cycles, as we will use these conserved quantities later in discussing a new family of integrable graphs.

Consider the directed cycle with $n$ vertices (see $G_\text{in}$ or $G_\text{out}$ in \cref{fig:GraphEmbedding}). Evripidou  et al. \cite{EKV22} refer to this as $\mathrm{KM}(n) \equiv B(n,1)$ for the work of Kac and van Moerbeke \cite{KM75} and Moser \cite{M74} whose work on the Volterra lattice predates the work of Kac and van Moerbeke. For reasons that will be clear in the sequel, let,
\begin{equation*}
    s = \left\lfloor\frac{n}{2}\right\rfloor + 1.
\end{equation*}
In the dynamical system generated from the graph, $\mathrm{KM}(n)$, there are always two conserved quantities,
\begin{equation*}
H_1 = \sum_{i} x_i \qquad \text{and} \qquad H_{s} = \prod_{i} x_i.
\end{equation*}
We note that the corresponding game (interaction) matrix has rank $n - 2$ when $n$ is even and rank $n - 1$ when $n$ is odd. Consequently, when $n$ is odd, there will be $\left\lceil\tfrac{n}{2}\right\rceil = \left\lfloor\tfrac{n}{2}\right\rfloor + 1$ conserved quantities and when $n$ is even, there will be $\tfrac{n}{2} + 1 = \left\lfloor\frac{n}{2}\right\rfloor + 1$ conserved quantities.
 
Define a \textit{non-edge} as the pair $(i,j)$ with $i < j$ so that neither $(i,j)$ nor $(j,i)$ is an element of $E$. Likewise, define a \textit{non-cycle} as a sequence $(i_1,\dots,i_r)$ so that, $i_j < i_{j+1}$, and for all pairs $(i_j, i_k)$, the pair $(i_j, i_k)$ is a non-edge and the pair $(i_1,i_r)$ is also a non-edge. Thus, $(i_1,\dots,i_r)$ forms a cycle in the graph complement of $\mathrm{KM}(n)$. Denote the set of all non-edges as $\overline{E}$ and the set of all non-cycles of length $l$ as $\overline{\mathcal{C}}(l)$. We can think of non-edges as non-cycles of length $2$, that is $\overline{E} = \overline{\mathcal{C}}(2)$. Notice the graph $\mathrm{KM}(n)$ posses non-empty sets $\overline{\mathcal{C}}(k)$ for $2 \leq k \leq \left\lfloor\tfrac{n}{2}\right\rfloor$.

Define 
\begin{equation}
    H_k = \sum_{c \in \overline{C}(k)} \prod_{j \in c} x_j.
    \label{eqn:ItohCC}
\end{equation}
The following result follows from \cite{I08} and is a restatement of the classical integrability of the Volterra lattice.
\begin{theorem} For $1 \leq k \leq s$, the quantity $H_k$ is conserved in the replicator dynamics generated by $\mathrm{KM}(n)$. Moreover, these conserved quantities $H_k$ commute under the nonlinear bracket in \cref{eqn:Bracket}.
\hfill\qed
\end{theorem} 

Both Itoh \cite{I87} and Paik and Griffin \cite{PG23} note that the conserved quantities for odd tournament graphs $B\left(n,\left\lfloor\tfrac{n}{2}\right\rfloor\right)$ can be written in terms of sums over products of cycles (rather than non-cycles). For other graphs in $B(n,k)$, the conserved quantities can be expressed as combinations of both sums over products of cycles and non-cycles. Consequently, all conserved quantities can be constructed directly out of the edge structure of $B(n,k)$ or its graph complement. (Details on this generalisation can be found in the SI.) This observation further supports our goal of characterising the dynamics of the replicator (Lotka-Volterra) equations directly from the graph structures used to generate them.

Given the work by Itoh et al. \cite{I87,I08}, we now make the following definition.
\begin{definition} An oriented directed graph $G$ admits Itoh-style conserved quantities if it has a set of algebraically independent conserved quantities that can be written in the form,
\begin{equation}
    H = \sum_{j\in \mathcal{J}}\prod_{i\in\mathcal{I}(j)} x_j,
    \label{eqn:ItohStyle}
\end{equation}
for appropriate index sets $\mathcal{I}(j)$ and $\mathcal{J}$.
\end{definition}

Importantly, these conserved quantities can be directly extracted from the graph structure, indicating a connection between dynamics and structure. Therefore, we could only identify graphs as integrable if they possessed these types of conserved quantities. 

We require one final, seminal, result proved by Evripidou  et al. \cite{EKV22}, relating graph decloning and integrability. 
\begin{theorem}[Proposition 4.5 of \cite{EKV22}] A graph $G_\mathbf{A}$ is (super) integrable if and only if its decloned form $\tilde{G}_{\tilde{\mathbf{A}}}$ is integrable. \hfill\qed
\label{thm:Declone}
\end{theorem}

\section{Skip-Vertex Graphs: A New Family of Integrable Graphs}\label{sec:SkipGraphs}
Intuitively, Bogoyavlenskij  graphs are cycles with edges added in a way that preserves a notion of rotational symmetry. Here we consider another instance where adding edges to a directed cycle produces another family of integrable graphs, but symmetry is broken.
\begin{definition}[Skip-Vertex Graph] The skip-vertex graph on $n \geq 4$ vertices with $k \leq \tfrac{n}{2}$ skips ($k < 2$ for $n = 4$) is (isomorphic to) a graph constructed by adding the edges $(1,3), (3,5),\dots,(2k-1,2k+1)$ to the directed cycle with $n$ vertices. We refer to these added edges as \textit{skip-edges} and denote this graph $\Sk(n,k)$. Note when $k = \tfrac{n}{2}$ addition is modulo $n$ so that the last added edge connects Vertex $n-1$ to Vertex 1. If the skip-edge $(2k-1,2k+1)$ is present, we call vertex $2k$ a skipped vertex.
\end{definition}
The skip-graph $\Sk(8,3)$ is shown in \cref{fig:SkipGraph}.
\begin{figure}[htbp]
\centering
\includegraphics[width=0.45\textwidth]{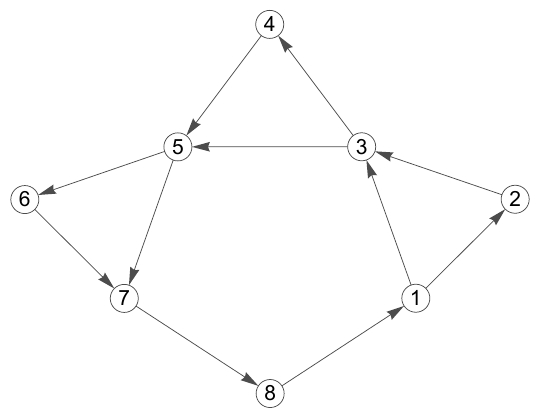}
\caption{The skip-graph $\Sk(8,3)$ is shown. Notice this skip-graph is not as symmetric as a graph in the family $B(n,k)$.}
\label{fig:SkipGraph}
\end{figure}
The graph $\Sk(n,k)$ has an interaction matrix with the following form,
\begin{equation}
    \mathbf{A} = \begin{bmatrix}
    0 & -1 & -\alpha_1 & 0 & 0 & \cdots & \alpha_{n-1} & 1\\
    1 & 0 & -1 & 0 & 0 & \cdots & 0 & 0\\
    \alpha_1 & 1 & 0 & -1 & -\alpha_3 & \cdots & 0 & 0\\
    0 & 0 & 1 & 0 & 0 & \cdots & 0 & 0 \\
    0 & 0 & \alpha_3 & 0 & 0 & \cdots & 0 & 0\\
    \vdots & \vdots & \vdots& \vdots & \vdots& \ddots & \vdots& \vdots\\
    -\alpha_{n-1} & 0 & 0 & 0 & 0 & \cdots & 0 & -1\\
    -1 & 0 & 0 & 0 & 0 & \cdots & 1 & 0
    \end{bmatrix},
    \label{eqn:AMatrixSkipGraph}
\end{equation}
where $\alpha_1 = 1$ and by extension, $\alpha_{2i-1} = 1$ if $k \geq i$ and otherwise $\alpha_{2i-1} = 0$. 

\begin{lemma} The dynamics generated by a skipped-vertex $\Sk(n,k)$ graph has a fixed point in the interior of $\Delta_{n-1}$ if and only if $n$ is even.
\label{lem:EvenSkipGraph}
\end{lemma}
\begin{proof} Suppose $n$ is even, and let $\mathbf{A}$ be the interaction matrix from a graph $\Sk(n,k)$. Interior fixed points of \cref{eqn:Replicator} arise from solutions to the equation, $\mathbf{A}\mathbf{x} = \mathbf{0}$, when $\mathbf{x}$ is constrained to lie in the interior of $\Delta_{n-1}$. Thus, it suffices to analyse the eigenvectors of $\mathbf{A}$ corresponding to the zero eigenvalues (and to the Casimirs of the system).

As with the directed cycle with $n$ vertices, $\mathbf{A}$ has two eigenvectors corresponding to zero eigenvalues with form,
\begin{equation*}
    \mathbf{v}_1 = \begin{bmatrix}0 \\1\\ 0 \\ 1\\ \vdots\\ 0 \\1\end{bmatrix} \quad \text{and} \quad
    \mathbf{v}_2 = \begin{bmatrix}1 \\ -\alpha_1 \\ 1 \\-\alpha_3 \\ \vdots \\ 1\\0
    \end{bmatrix},
\end{equation*}
with alternating structure. This can be shown by direct computation using \cref{eqn:AMatrixSkipGraph}. Notice, if $\alpha_{2i-1} = 0$ for all $i$, then these are the two eigenvectors corresponding to eigenvalue zero in the directed cycle. Let, $\mathbf{v} = \alpha\mathbf{v}_1 + \beta\mathbf{v}_2$, then $\mathbf{A}\mathbf{v} = \mathbf{0}$. We will show there are $\alpha$ and $\beta$ so that $\mathbf{v} \in \Delta_{n-1}$. If $\mathbf{v} = \langle{v_1,\dots,v_n}\rangle$,
then,
\begin{equation*}
    \sum_i v_i = \frac{n}{2}\alpha + \left(\frac{n}{2} - k\right)\beta.
\end{equation*}
If $\mathbf{v} \in \Delta_{n-1}$, then we require
\begin{equation*}
    \frac{n}{2}\alpha + \left(\frac{n}{2} - k\right)\beta = 1,
\end{equation*}
which implies,
\begin{equation*}
    \beta = \frac{n\alpha - 2}{2k - n}.
\end{equation*}
We also require $\mathbf{v} > 0$. This yields two inequalities,
\begin{equation*}
    \frac{\alpha n-2}{n-2k} > 0 \quad \text{and} \quad \alpha + \frac{\alpha n - 2}{n-2k} > 0.
\end{equation*}
Using our assumption that $n > 2k$, it follows that we require
\begin{equation*}
    \frac{1}{n-k} < \alpha < \frac{2}{n},
\end{equation*}
for $\mathbf{v}$ to be a vector in $\Delta_{n-1}$. Thus, the dynamics generated by a graph in $\Sk(n,k)$ have an interior fixed point when $n$ is even.

Now suppose $n$ is odd. As in the case of an odd cycle, the interaction matrix has only one eigenvector with eigenvalue 0. This vector has form,
\begin{equation*}
    \mathbf{v} = \begin{bmatrix}
    1\\
    1-\alpha_1\\
    1\\
    1-\alpha_3\\
    1\\
    \vdots\\
    1 - \alpha_{2k-1}\\
    1
    \end{bmatrix}
\end{equation*}
Notice if $\alpha_i = 0$, we recover the unique eigenvector of the interaction matrix of the directed cycle corresponding to the zero eigenvalue, as expected. This vector cannot be scaled to lie in the interior of $\Delta_{n-1}$. Consequently, there is no interior point solution and all solutions asymptotically decay to a lower-order system solution.
\end{proof}

From the preceeding proof, we see that the rank of the interaction matrix of a skip graph with $2n$ vertices is at most $2n - 2$. The structure of the two eigenvectors identified in the proof imply that through Gaussian elminiation one can remove the final two rows and columns of the interaction matrix leaving a reduced $(2n-2) \times (2n-2)$ matrix,
\begin{equation*}
    \tilde{\mathbf{A}} = 
\begin{bmatrix}
    0 & -1 & -\alpha_1 & 0 & 0 & \cdots &0 & 0\\
    1 & 0 & -1 & 0 & 0 & \cdots & 0 &  0\\
    \alpha_1 & 1 & 0 & -1 & -\alpha_3 & \cdots & 0 & 0\\
    0 & 0 & 1 & 0 & 0 & \cdots &  0 & 0\\
    0 & 0 & \alpha_3 & 0 & 0 & \cdots & 0 & 0\\
    \vdots & \vdots & \vdots& \vdots & \vdots& \ddots & \vdots & \vdots\\
    0 & 0 & 0 & 0 & 0 & \cdots & 0 & -1\\
    0 & 0 & 0 & 0 & 0 & \cdots & 1 & 0
    \end{bmatrix},
\end{equation*}
that is the interaction matrix of a directed acyclic graph. This is a banded skew-symmetric matrix. Direct computation of the Pfaffian \cite{C1849,B04} for this matrix shows that it has value $(-1)^{n-1}$ (just as in the case for a directed graph). Consequently, we have $\det(\tilde{\mathbf{A}}) = (\pm 1)^2$, implying this reduced matrix has full rank. It follows that the rank of the interaction matrix produced by a skip-vertex graph with $2n$ vertices must be $2n - 2$. This is sufficient to show that Poisson manifold generated by this graph has rank $2n-2$ \cite{EKV22}. Note, it is also possible to directly compute the determinant of $\tilde{\mathbf{A}}$, which is always $1$.

We now construct conserved quantities for a skip-graph with $2n$ vertices. Let $V_s$ denote the set of vertices that are not skipped. Notice the eigenvalues in the proof provide two conserved quantities for skip graphs,
\begin{align}
    &H_{n} = x_2x_4\cdots x_{2n}\label{eqn:SkipCCn}\\
    &H_{n+1} = \prod_{i \in V_s} x_i.\label{eqn:SkipCCn1}
\end{align}
This can be seen immediately by reading these terms from $\mathbf{v}_1$ and $\mathbf{v}$ in the proof. For example, these two conserved quantities for the graph shown in \cref{fig:SkipGraph}, i.e., $\Sk(8,3)$ are,
\begin{align*}
    &H_4 = x_2x_4x_6x_8\\
    &H_5 = x_1x_3x_5x_7x_8.
\end{align*}
As always, $H_1 = x_1 + x_2 + \cdots + x_{2n}$ is a conserved quantity. 

Given the result in \cref{lem:EvenSkipGraph}, we consider only skip-graphs $\Sk(n,k)$ where $n$ is even and use an imagined contrivance of a fluid flowing along the edges of a directed cycle or skip-graph. The fact that the graph is strongly connected implies that this imaginary fluid must circulate through the graph.

Consider the directed cycle $\mathrm{KM}(n)$ with corresponding matrix $\mathbf{A}_0$. Consider an arbitrary non-cycle in $\mathrm{KM}(n)$ corresponding to the term $T = x_{i_1}\cdots x_{i_l}$. Then, computing $\dot{T}$ with respect to the dynamics generated by $\mathrm{KM}(n)$ yields,
\begin{multline*}
    \dot{T} = \left(\prod_{j=1}^n x_{i_j}\right)\left(\sum_{j=1}^n \mathbf{e}_{i_j}^T\right)\mathbf{A}_0\mathbf{x} = \\
    T \cdot (x_{i_1+1} - x_{i_1-1} + x_{i_2+1} - x_{i_2-1} + ... + x_{i_n+1} - x_{i_n-1}).
\end{multline*}
Addition in the index terms is taken modulo $n$ with appropriate adjustments to count from $1,\dots,n$ if needed. Notice, the sum on the right-hand-side performs in/out flow counting with respect to the non-cycle; i.e., the net sum is composed of terms corresponding to outward flow being negative and inward flow being positive. This is not the case if $T$ contains pairs of variables corresponding to an edge in the graph. Consequently, summing over all these inward and outward flows in non-cycles must yield zero because the graph is strongly connected and flow can only circulate. This is just a restatement of Itoh's \cite{I08} result.

Now decorate the time derivative operator with the graph used to formulate it. That is,
\begin{multline}
    D_{\mathrm{KM}(n)}[T] = \left(\prod_{j=1}^n x_{i_j}\right)\left(\sum_{j=1}^n \mathbf{e}_{i_j}^T\right)\mathbf{A}_0\mathbf{x} = \\
    T \cdot (x_{i_1+1} - x_{i_1-1} + x_{i_2+1} - x_{i_2-1} + ... + x_{i_n+1} - x_{i_n-1}).
    \label{eqn:TimeDerivSkipProof}
\end{multline}
Passing to $\Sk(n,k)$, the terms corresponding to non-cycles that do not contain vertices in skip-edges have no change to their time derivative. Now consider a skip-edge of form $(i_j,i_{j}+2)$. Say a non-cycle term $T$ contains $x_{i_j}$ (but not $x_{i_j+2}$). Then, 
\begin{equation*}
    D_{\Sk(n, k)}[T] = D_{\mathrm{KM}(n)}[T] - Tx_{i_j+2}.
\end{equation*}
Likewise, if $T$ contains $x_{i_j+2}$, then,
\begin{equation*}
    D_{\Sk(n, k)}[T] = D_{\mathrm{KM}(n)}[T] + Tx_{i_j+1},
\end{equation*}
This alteration occurs because there may be fluid flow cancellation (both in and out) to the elements of a non-cycle. However, the fact that the time-derivatives of the non-cycles are still counting fluid flow implies that the sum of all non-cycles of a given length will be a conserved quantity because of our imaginary fluid conservation. Thus, we state the following lemma, which follows from the discussion above. 
\begin{lemma} Let $\overline{C}_{\Sk(n,k)}(m)$ denote the non-cycles of the skip-graph $\Sk(n,k)$ of length $m$, and let,
\begin{equation*}
    H_m = \sum_{c \in \overline{C}_{\Sk}(m)} \prod_{j \in c} x_j.
    \label{eqn:SkipCC}
\end{equation*} 
Then $H_m$ is a conserved quantity for $\Sk(n,k)$. \hfill \qed
\end{lemma}
Notice that $2 \leq m \leq \tfrac{n}{2}$. Moreover, non-cycles of length $\tfrac{n}{2}$ correspond to the eigenvector $\mathbf{v}_1$ from the proof of \cref{lem:EvenSkipGraph}. Thus, for a skip-graph $\Sk(n,k)$ we have $\tfrac{n}{2} + 1$ conserved quantities, where $n$ is even. It is worth noting that this explains why conserved quantities are generated from both cycles and non-cycles in graphs in class $B(n,k)$ for large enough $k$. When there are no non-cycles available to use to create balanced flow, the conserved quantities are written in terms of cycles, which conserve flow differently. This is most obvious in the conserved quantities for balanced tournament graphs (see \cite{I87} or \cite{PG23}).

Using a similar argument, or adapting Itoh's proof technique, we also have the following.

\begin{lemma} If $H_1,\dots,H_{n+1}$ are the conserved quantities for $\Sk(n,k)$ (where $n$ is assumed to be even), then these quantities commute under the action of the nonlinear bracket in \cref{eqn:Bracket}. \hfill\qed.
\end{lemma}

It remains to argue that the conserved quantities we have found are algebraically independent. Consider a skip-graph with $\Sk(2n,k)$ with $2n$ vertices. The Jacobian matrix of the set of conserved quantities of this graph is,
\begin{equation*}
    \mathbf{J}(x_1,\dots,x_{2n}) = \begin{bmatrix} 
    \frac{\partial H_1}{\partial x_1} & \cdots & \frac{\partial H_1}{\partial x_{2n}}\\
    \vdots & \ddots & \vdots\\
    \frac{\partial H_n}{\partial x_1} & \cdots & \frac{\partial H_n}{\partial x_{2n}}\\
    \frac{\partial H_{n+1}}{\partial x_1} & \cdots & \frac{\partial H_{n+1}}{\partial x_{2n}}.
    \end{bmatrix}
\end{equation*}
This matrix has algebraic form,
\begin{equation*}
    \mathbf{M} = \begin{bmatrix}
        1 & 1 & \cdots & 1\\
        p_{1,1}(\mathbf{x}) & p_{1,2}(\mathbf{x}) & \cdots & p_{1,2n}(\mathbf{x})\\
        p_{2,1}(\mathbf{x}) & p_{2,2}(\mathbf{x}) & \cdots & p_{2,2n}(\mathbf{x})\\
        \vdots & \vdots & \ddots & \vdots\\
        p_{n-1,1}(\mathbf{x}) & p_{n-1,2}(\mathbf{x}) &\cdots & p_{n-1,2n}(\mathbf{x})\\
        p_{2n-(k+1),1}(\mathbf{x}) & p_{2n-(k+1),2}(\mathbf{x}) & \cdots & p_{2n-(k+1),2n}(\mathbf{x})
    \end{bmatrix},
\end{equation*}

where $p_{i,j}$ denotes either $0$ or a polynomial of degree $i$ in column $j$. Notice that each row has polynomials of distinct degree. Moreover, these polynomials cannot be multiples of each other because they were constructed from sums of non-cycles in the graph.

We can use elementary row operations to convert the first column to a unit vector without changing the degrees of each row. We then divide the second row by a degree-one polynomial and continue row reduction. The result will be a matrix with rational expressions as entries, but with each row having rational entries with distinct degrees in each row. This process can be repeated until the first $n+1$ columns have been converted to identity matrix columns, with no rows transforming to the zero-vector precisely because of the uniqueness of the degrees in each row and the construction of the polynomials in $\mathbf{M}$ from (derivatives of) sums of non-cycles of varying size or derivatives of the Casimir generated from the cycle of skip vertices (see \cref{eqn:SkipCCn,eqn:SkipCCn1}). From this construction, we see that $\mathbf{M}$ must have full rank. From the Jacobian criterion \cite[Theorem 2.3]{P08} and \cref{def:LAIntegrability}, we have the following result.

\begin{theorem} Let $G = \Sk(n,k)$ for even $n$ and $2k < n$. Then $G$ is integrable. \hfill\qed
\label{thm:SkipVertex}
\end{theorem}

Illustration of this result for the skip graphs containing six vertices is provided in a Mathematica notebook in the supplemental information. These graphs are used in the taxonomy in \cref{sec:Taxonomy}.

\section{Integrability of Graph Embeddings}\label{sec:Embeddings}
We now proceed to generalise the work by Paik and Griffin \cite{PG23} related to Graph-to-Vertex Embeddings.

\begin{theorem} Let $\Gout = (\Vout,\Eout)$ and $\Gin = (\Vin,\Ein)$ be two graphs that admit Itoh-style conserved quantities and produce integrable dynamical systems with these conserved quantities in the sense of \cref{def:LAIntegrability}. If 
$J = \Gin \hookrightarrow_v \Gout$, then $J$ is integrable.
\label{thm:Embed}
\end{theorem}
\begin{proof}
Without loss of generality, we assume $v=1$ and that $\Gout$ has $\nout$ vertices and $\Gin$ has $\nin$ vertices. Denote the interaction matrices of the graphs by $\Aout$ and $\Ain$, respectively. The interaction matrix of $J$ is given by \cref{eqn:Aembed} as,
\begin{equation*}
    \mathbf{A} = \begin{bmatrix} \Ain & \mathbf{R}^T\\-\mathbf{R} & \AoutTilde
    \end{bmatrix}.
\end{equation*}
Define the variables vectors,
\begin{align*}
    &\xin = \langle{x_1,\dots,x_{\nin}}\rangle,\\
    &\xout = \langle{x_{\bar{1}},\dots,x_{\noutTilde}}\rangle,\quad \text{and} \\
    &\xoutTilde = \langle{x_{\bar{2}},\dots, x_{\noutTilde}}\rangle.
\end{align*} 

Suppose $\Gout$ has $\ssout$ conserved quantities denoted $H_1,\dots,H_{\ssout}$ and suppose $\Gin$ has $\ssin$ conserved quantities denoted, $F_1,\dots,F_{\ssin}$. We will assume the Hamiltonians are,
\begin{align*}
    &H_1(x_{\bar{1}},\dots,x_{\noutTilde}) = x_{\bar{1}} + \cdots + x_{\noutTilde}\quad \text{and}\\
    &F_1(x_1,\dots,x_{\nin}) = x_1 + \cdots + x_{\nin}.
\end{align*}
Necessarily, $\Gout$ has at most one Casimir containing $x_1$ with unit power. For simplicity, denote this Casimir as $C$ with corresponding eigenvector for $\Aout$, $\bm{\chi} = \langle{\chi_{1}, \dots\chi_{\nout}}\rangle$. Our assumption that $\Gout$ admits Itoh-style conserved quantities implies that $\bm{\chi}$ is made up of ones and zeros with $\chi_1 = 1$. This function $C$ is among the functions $H_1,\dots,H_{\nout}$ and there may be other Casimirs not containing $x_1$.

Our assumption that $\Gin$ and $\Gout$ admit Itoh-style conserved quantities implies that we can write each of the conserved quantities of $\Gout$ and $\Gin$ as,
\begin{align} 
&H_i(x_{\bar{1}},\xoutTilde) = \sum_{l \in \mathcal{U}_i} x_{\bar{1}}P_l(\xoutTilde) + \sum_{l \in \mathcal{V}_i}  Q_l(\xoutTilde)\label{eqn:GoutConserved}\\
&F_i(\xin) = \sum_{l \in \mathcal{W}_i} A_l(\xin),\label{eqn:GinConserved}
\end{align}
where $P_l$, $Q_l$ and $A_l$ are monomials and $\mathcal{U}_i$, $\mathcal{V}_i$ and $\mathcal{W}_i$ are appropriately defined index sets over the terms of the conserved quantities $H_i$ and $F_i$. See \cref{eqn:ItohStyle}. Also, an example is given in \cref{eqn:ItohCC} for the cycle graphs. In particular, the Casimir $C$ will have form $C = x_1P(\xoutTilde)$. 

The following procedure can be used to construct the new conserved quantities for $J$:
\begin{enumerate}
    \item Using the Casimir $C$ of $\Gout$ containing the variable $x_1$ and the non-Hamiltonian conserved quantities of $\Gin$ define,
    \begin{equation}
        \Gamma_{i}(\xin,\xoutTilde) = F_i(\xin)^{\chi_{1}}P(\xoutTilde)^r,
        \label{eqn:Gamma}
    \end{equation}
    for $i \in \{2,\dots,\ssin\}$. Here $r$ is the order of $F_i$ as a polynomial. Note that we have replaced $x_1$ in $C$ with $F_i(\xin)$. Recall by assumption $\chi_{1} = 1$, but we leave it explicit for clarity in the proof. We show that when $F_i$ is a Casimir of $\Gin$, then $\Gamma_{i}$ is a Casimir of $J$.
    
    \item Using the Hamiltonian $F_1(\xin)$ of $\Gin$ and every conserved quantity of $\Gout$ that contains $x_1$ define,
    \begin{equation}
        Z_{i} = \sum_{l \in \mathcal{U}_i} F_1(\xin)P_l(\xoutTilde) + \sum_{l \in \mathcal{V}_i}  Q_l(\xoutTilde).
        \label{eqn:Z}
    \end{equation}
    That is, we replace $x_1$ in $H_i$ with $F_1(\xin)$. Any other remaining conserved quantities that do not contain $x_1$ are carried forward to $J$, which can also be thought of as replacing $x_1$ with $F_1$ in a trivial way.
\end{enumerate}
The procedure is illustrated in \cref{fig:Proof-Pic}.
\begin{figure}
\centering
\includegraphics[width=0.75\textwidth]{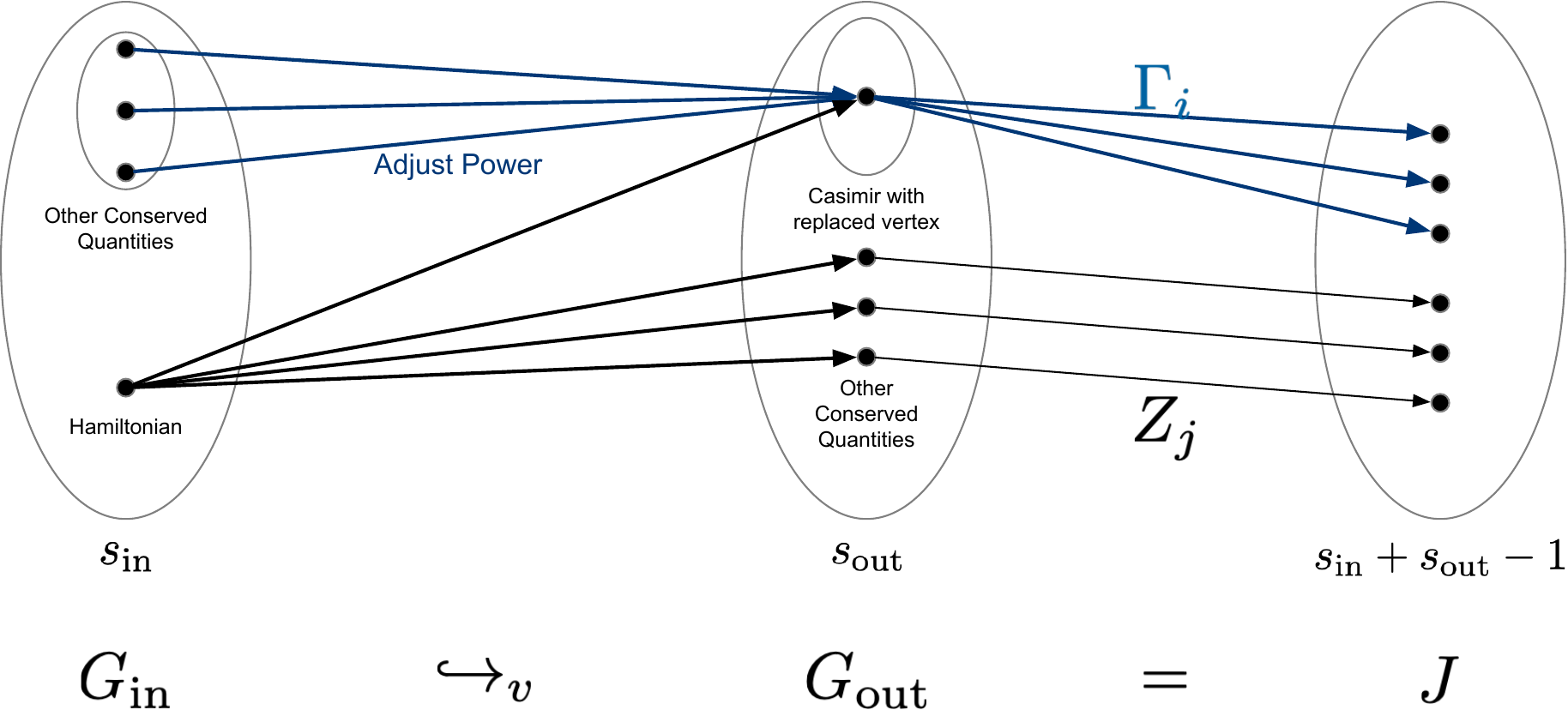}
\caption{To construct conserved quantities for an embedding of two Bogoyavlenskij  graphs, first replace all instances of the variable representing the embedding vertex in the Casimirs of the outer graph with the conserved quantities from the inner graph, correcting the power of the remaining terms appropriately. Then replace all instances of the variable representing the embedding vertex with the Hamiltonian of the inner graph in non-Casimir conserved quantities of the outer graph.}
\label{fig:Proof-Pic}
\end{figure}

\begin{claim} The quantities defined in \cref{eqn:Gamma} are conserved.
\end{claim}
\begin{proof} By assumption, there are vectors $\bm{\alpha}_l$ so that,
\begin{equation*}
    \frac{dA_l}{dt} = A_l(\xin)\bm{\alpha}_l^T \Ain\xin,
\end{equation*}
with,
\begin{equation}
    \frac{d F_i}{dt} = \sum_{l \in \mathcal{W}_i} A_l(\xin)\bm{\alpha}_l^T \Ain\xin = 0.
    \label{eqn:InGoesZero}
\end{equation}
Likewise, for the eigenvector vector $\bm{\chi} = \langle{\chi_{1},\dots\chi_{\nout}}\rangle$, we have,
\begin{equation*}
    \frac{d C}{dt} = C \bm{\chi}^T\Aout\xout  = 0.
\end{equation*}
The vectors $\bm{\alpha}_l$ and $\bm{\chi}_j$ are just the powers appearing in the monomials $A_l$ and Casimir $C_j$. In particular, $\bm{\chi}$ is an eigenvector of $\Aout$. To be explicit, let $\bar{\bm{\chi}} =\langle{\chi_{2},\dots\chi_{\nout}}\rangle$. We know from, \cref{eqn:Aout}, that,
\begin{equation*}
    \bm{\chi}^T\Aout = \begin{bmatrix}\chi_{1} & \bar{\bm{\chi}}^T\end{bmatrix} \begin{bmatrix} 0 & \mathbf{r}^T\\-\mathbf{r} & \AoutTilde\end{bmatrix} = \begin{bmatrix} -\bar{\bm{\chi}}^T\mathbf{r} & \chi_{1}\mathbf{r}^T + \bar{\bm{\chi}}^T\AoutTilde\end{bmatrix} = \mathbf{0}.
\end{equation*}
Thus, we have the identities, 
\begin{align}
    &\bar{\bm{\chi}}^T\mathbf{r} = 0 \label{eqn:OutGoesZero1}\\
    &\chi_{1}\mathbf{r}^T + \bar{\bm{\chi}}^T\AoutTilde = \mathbf{0},\label{eqn:OutGoesZero2}
\end{align}
which can be interpreted as a kind of flow balance equations. By construction, $\Gamma_{i} = F_i(\xin)P(\xoutTilde)$ is just a sum of monomials. To account for the embedding and to be notationally consistent, let $\tilde{\bm{\chi}}_j =\langle{\chi_{\tilde{2}},\dots\chi_{\noutTilde}}\rangle$. 

Consequently, there are vectors 
$\bm{\gamma}_{l} = \langle{\chi_{1}\bm{\alpha}_l,r\tilde{\bm{\chi}}}\rangle$ so that,
\begin{multline*}
    \frac{d\Gamma_{i} }{dt} = \sum_{l\in \mathcal{W}_i}A_lP^r \bm{\gamma}_{l}^T\mathbf{A}\mathbf{x} = \\
    \sum_{l \in \mathcal{W}_i} A_lP^r \left[ 
    \left(\chi_{1}\bm{\alpha}_l^T\Ain - r\tilde{\bm{\chi}}^T\mathbf{R}\right)\xin + 
    \left(\chi_{1}\bm{\alpha}_l^T\mathbf{R}^T + r\tilde{\bm{\chi}}^T\AoutTilde\right)\xoutTilde
    \right].
\end{multline*}
Recall, $r$ is the order of $F_i$. From the structure of $\mathbf{R}$, see \cref{eqn:R}, we see that,
\begin{equation}
    \bm{\alpha}_l^T\mathbf{R}^T = r\mathbf{r}^T.
    \label{eqn:rFactor}
\end{equation}
Likewise, we have from \cref{eqn:OutGoesZero1} we see,
\begin{equation}
    \tilde{\bm{\chi}}^T\mathbf{R} = \mathbf{0}
    \label{eqn:FactorChi}
\end{equation}
Therefore, from \cref{eqn:InGoesZero,eqn:OutGoesZero1,eqn:OutGoesZero2,eqn:rFactor,eqn:FactorChi} we have,
\begin{multline*}
    \frac{d\Gamma_{i} }{dt} = 
    P^r\sum_{l \in \mathcal{W}_i}A_l\bm{\alpha}_l^T\Ain\xin - rP^r\sum_{l \in \mathcal{W}_i}A_l\tilde{\bm{\chi}}^T\mathbf{R}\xin + \\P^r\sum_{l\in \mathcal{W}_i} rA_l\left(\chi_{1}\mathbf{r}^T + \bm{\chi}^T\AoutTilde\right)\xoutTilde = 0.
\end{multline*}
\end{proof}
From the proof of this claim, it's immediately clear that if $F_i$ is a Casimir of $\Gin$, then there is only one term in the sum and the resulting vector $\bm{\gamma}_{i}$ must be an eigenvector of $\mathbf{A}$ with eigenvalue $0$ and therefore $\Gamma_{i}$ is a Casimir of the new system.

\begin{claim} The quantities defined in \cref{eqn:Z} are conserved.
\end{claim}
\begin{proof} This follows from the use of the Hamiltonian $F_1(\xin)$ in \cref{eqn:Z} and the fact that $x_{\bar{1}}$ is being replaced by $x_1,\dots,x_{\nin}$ and $Z_i$ is constructed from the conserved quantity in \cref{eqn:GoutConserved}. That is, we are simply renaming $x_{\bar{1}}$ and this will not affect the time derivative of the quantity. 
\end{proof}

\begin{claim} The constructed conserved quantities commute under the action of the quadratic bracket in \cref{eqn:Bracket}.
\end{claim}
\begin{proof}
It suffices to consider two functions of form,
\begin{equation*}
    \Gamma_\alpha = F_\alpha(\xin)P_\alpha(\xoutTilde)^{r_\alpha} \quad \Gamma_\beta = F_\beta(\xin)P_\beta(\xoutTilde)^{r_\beta},
\end{equation*}
under the assumption that $F_\alpha$ and $F_\beta$ commute under the bracket restricted to $\xin$ and $x_{\bar{1}}^{r_\alpha}P_\alpha(\xoutTilde)^{r_\alpha}$ and $x_{\bar{1}}^{r_\beta}P_\beta(\xoutTilde)^{r_\beta}$ commute under the bracket restricted to $\xout$. Here $r_\alpha$ is the order of the polynomial $F_\alpha$ and $r_\beta$ is the order of the polynomial $F_\beta$.

In what follows, we assume that $1 < \cdots < \nin < \bar{2} < \cdots \noutTilde$. That is, the bar indicates an appropriate index shift as a result of the embedding. Consider the bracket,
\begin{equation*}
    \left\{ \Gamma_\alpha, \Gamma_\beta\right\}_\mathbf{A} = 
    \sum_{i < j} \underbrace{A_{ij}x_ix_j \left(
    \frac{\partial F_\alpha P_\alpha^{r_\alpha}}{\partial x_i}\frac{\partial F_\beta P_\beta^{r_\beta}}{\partial x_j} - 
    \frac{\partial F_\alpha P_\alpha^{r_\alpha}}{\partial x_j}\frac{\partial F_\beta P_\beta^{r_\beta}}{\partial x_i}
    \right)}_{\Xi}
\end{equation*}
Denote the term under the sum by $\Xi$. When $i< j \leq \nin$, then $\Xi$ simplifies to,
\begin{equation*}
    P_\alpha^{r_\alpha}P_\beta^{r_\beta}A_{ij}x_ix_j\left(\frac{\partial F_\alpha }{\partial x_i}\frac{\partial F_\beta }{\partial x_j} - 
    \frac{\partial F_\alpha }{\partial x_j}\frac{\partial F_\beta}{\partial x_i}\right),
\end{equation*}
which vanishes under the sum because $F_\alpha$ and $F_\beta$ commute. 

Now consider the remaining elements of the sum with $j \geq \bar{2}$. The assumption that the graphs admit Itoh-style conserved quantities is now critical. Assume $\Xi$ does not evaluate to zero. If $i \geq \bar{2}$, then both derivatives affect the terms $P_\alpha^{r_\alpha}$ and $P_\beta^{r_\beta}$. We have,
\begin{multline*}
    \Xi = 
    A_{ij}x_ix_j \left(
    \frac{\partial F_\alpha P_\alpha^{r_\alpha}}{\partial x_i}\frac{\partial F_\beta P_\beta^{r_\beta}}{\partial x_j} - 
    \frac{\partial F_\alpha P_\alpha^{r_\alpha}}{\partial x_j}\frac{\partial F_\beta P_\beta^{r_\beta}}{\partial x_i}
    \right) = \\
    A_{ij}\left(r_\alpha r_\beta F_\alpha F_\beta P_\alpha^{r_\alpha} P_\beta^{r_\beta} - r_\alpha r_\beta F_\alpha F_\beta P_\alpha^{r_\alpha} P_\beta^{r_\beta}\right) = 0,
\end{multline*}
because the $x_ix_j$ multiplier effectively removes the differentiation in the functions $P_\alpha^{r_\alpha}$ and $P_\beta^{r_\beta}$. 

This leaves the case when $i \leq \nin$ and $j \geq \bar{2}$. Here we have,
\begin{equation*}
    \Xi = \underbrace{r_\beta A_{ij}F_\alpha F_\beta P_\alpha^{r_\alpha}P_\beta^{r_\beta}}_{T_1} - \underbrace{r_\alpha A_{ij}F_\alpha F_\beta P_\alpha^{r_\alpha}P_\beta^{r_\beta}}_{T_2}
\end{equation*}
However, because the order of $F_\alpha$ is $r_\alpha$, the term $T_1$ will occur $r_\alpha$ times in the sum, while the term $T_2$ will occur $r_\beta$ times in the sum because $F_\beta$ has order $r_\beta$. Thus, when summing over all $i$ indices occurring in $F_\alpha$ and $F_\beta$, the sum will produce,
\begin{equation*}
    r_\alpha r_\beta A_{ij}F_\alpha F_\beta P_\alpha^{r_\alpha}P_\beta^{r_\beta} - r_\beta r_\alpha A_{ij}F_\alpha F_\beta P_\alpha^{r_\alpha}P_\beta^{r_\beta} = 0.
\end{equation*}
It follows immediately that any pair of constructed conserved quantities commute under the quadratic bracket.
\end{proof}

We now show that there are a sufficient number of conserved quantities to ensure that the resulting dynamics generated by the graph $J$ are integrable. Let $\ccin$ and $\ccout$ be the total number of Casimirs in $\Gin$ and $\Gout$ respectively. By construction, $J$ has at least $\ccin + \ccout - 1$ Casimirs consisting of the $\ccin$ $\Gamma_i$ that were constructed and the remaining $\ccout - 1$ Casimirs that do not contain $x_1$. However, it follows from the structure of $\mathbf{A}$ that these are the only possible Casimirs, as any other Casimirs would have the form $F_i(\xin)Q(\xout)$, where $Q$ is yet another Casimir of $\Gout$ containing $x_1$. Consequently, we have:
\begin{multline*}
    \mathrm{Rank}(\mathbf{A}) = (\nin + \nout - 1) - (\ccin + \ccout -1) = (\nin - \ccin) + (\nout - \ccout) =\\
    \mathrm{Rank}(\Ain) + \mathrm{Rank}(\Aout).
\end{multline*}
As illustrated in \cref{fig:Proof-Pic}, we have constructed $\ssin + \ssout - 1$ conserved quantities. From \cref{def:LAIntegrability} and our assumption on the integrability of $\Gin$ and $\Gout$ we know that,
\begin{align*}
    &\ssin + \frac{1}{2}\mathrm{Rank}(\Ain) = \nin\\
    &\ssout + \frac{1}{2}\mathrm{Rank}(\Aout) = \nout.
\end{align*}
Adding these quantities together and subtracting $1$ yields,
\begin{equation*}
    \ssin + \ssout - 1 + \frac{1}{2}\mathrm{Rank}(\mathbf{A}) = \nin + \nout - 1.
\end{equation*}

\begin{claim} The constructed conserved quantities are algebraically independent.
\end{claim}
\begin{proof} Consider the conserved quantities defined in \cref{eqn:Gamma,eqn:Z}. These are constructed from algebraically independent conserved quantities written in terms of distinct sets of variables and produced by multiplication (and simple powers). Consequently, if there is an algebraic function $\mathfrak{F}(\Gamma_2,\dots,\Gamma_{\ssin},Z_1,\dots,Z_{\ssout}) = 0$, then this function must operate on the classes of conserved quantities independently. That is, it follows that neither the conserved quantities $F_1,\dots F_{\ssin}$ nor the conserved quantities $H_1,\dots,H_{\ssout}$ are independent because $\mathfrak{F}$ can be decomposed into functions operating on the two classes of variables separately. This contradicts our assumption that $F_1,\dots F_{\ssin}$ and $H_1,\dots,H_{\ssout}$ are both independent.
\end{proof}

It follows at once that the dynamics generated by $J$ are integrable. This completes the proof.
\end{proof}

Interestingly, cloning is the simplest example of this process. If we take the unconnected graph of $n$ vertices and embed it into a single vertex, then we would obtain a cloning graph with n cloned vertices. The simplest example of this process for connected graphs is the embedding of a directed three cycle into another directed three cycle, which is used as an ecological model by Allesina et al. \cite{AL11} and shown to be integrable by Paik and Griffin \cite{PG23} in their discussion on the integrability of tournament embeddings, which is generalised by \cref{thm:Embed}. Instead, we illustrate the more complex case by embedding a directed three-cycle into a directed four-cycle, as illustrated in \cref{fig:GraphEmbedding}. The conserved quantities for $\Gin$ are,
\begin{equation*}
    F_1 = x_1 + x_2 + x_3 \quad \text{and} \quad F_2 = x_1x_2x_3.
\end{equation*}
The conserved quantities for $\Gout$ are,
\begin{equation*}
    H_1 = x_1 + x_2 + x_3 + x_4, \quad H_2 = x_1x_3,\quad \text{and} \quad H_3 = x_2x_4.
\end{equation*}
Notice, $\Gout$ has two Casimirs, $H_2$ and $H_3$, only one of which contains $x_1$. We have,
\begin{equation*}
    \Gamma_1 = \underbrace{(x_1x_2x_3)}_{F_2}\underbrace{x_{\bar{3}}^3}_{P^3},
\end{equation*}
and
\begin{align}
    &Z_1 = x_1 + x_2 + x_3 + x_{\bar{2}} + x_{\bar{3}} + x_{\bar{4}}\\
    &Z_2 = (x_1 + x_2 + x_3)x_{\bar{3}}\\
    &Z_3 = x_2x_4
\end{align}
Direct inspection shows that these are conserved quantities under the replicator dynamics generated by the graph shown in \cref{fig:GraphEmbedding} and that they are algebraically independent and commute under the action of the bracket.

Before concluding this section, it is worth noting that the matrix of the embedded graph need not have only $0,\pm1$ entries. The matrix, corresponding to $\Ain$ can be multiplied by any constant, effectively speeding up the dynamics in this subpopulation of variables, without changing the integrability of the system. This allows us to extend our results to a broader class of interaction matrices and potential make formal statements about the integrability of dynamics that exhibit multiple, hierarchical time scales. We omit a formal analysis of this observation and leave it for future work.

\section{Multiple Simultaneous Embeddings}\label{sec:MultipleEmbeddings}
Paik and Griffin note in \cite{PG23} that this embedding procedure can be generalised to multiple simultaneous embeddings to produce integrable (tournament) graphs. We now assume that multiple graphs with Itoh-style conserved quantities are being embedded into a single graph also with Itoh-style conserved quantities. For this, we use the notation,
\begin{equation}
J = (\Gin^1,\dots,\Gin^m) \hookrightarrow_{(v_1,\dots,v_m)} \Gout.
\end{equation}
For simplicity, we now refer to the conserved quantities of graph $G^j_\text{in}$ (with $j \in \{1,\dots,m\}$) as $F^j_i$ with $i \in \{1,\dots,\ssin^j$\} and the conserved quantities of $\Gout$ by $H_i$ for $i\in \{1,\dots\ssout\}$. Without loss of generality, we assume that the vertices into which we are embedding are ordered so that $v_i = i$ for $i\in\{1,\dots,m\}$. By assumption, the conserved quantities of $\Gout$ can be written as,
\begin{equation*}
    H_i = \sum_{l \in \mathcal{U}_i} B_l(x_1,\dots, x_{\nout}),
\end{equation*}
where $B_l$ is a monomial (product) and $\mathcal{U}_i$ is an appropriate index set (as before). Also, any Casimir $C_i$ of $\Gout$ containing $x_{i_1},\dots,x_{i_k}$ for $1 \leq i_1 < \cdots < i_k \leq m$ can be written,
\begin{equation*}
    C_i(\xout) = x_{i_1}\cdots x_{i_k} P_i^{i_1,\dots, i_k}(\xoutTilde),
\end{equation*}
where, the definition of $\xoutTilde$ is modified appropriately and $P_i^{i_1,\dots,i_k}$ is an appropriate monomial.

To build the conserved quantities of $J = (\Gin^1,\dots,\Gin^m) \hookrightarrow_{(v_1,\dots,v_m)} \Gout$, we use the following procedure, which generalises the one illustrated in \cref{fig:Proof-Pic}.
\begin{enumerate}
    \item Suppose $C_i(\xout) = x_{i_1}\cdots x_{i_k} P(\xoutTilde)$ is a Casimir of $\Gout$ containing $x_{i_1},\dots x_{i_k}$, where $i_1,\dots,i_k$ are a subset of the indices into which $(\Gin^1,\dots,\Gin^m)$ are embedded. Let $H^{i_1}_{j_1},\dots,H^{i_k}_{j_k}$ be conserved quantities of $G^{i_1},\dots,G^{i_k}$ with orders $r_{i_1},\dots,r_{i_k}$. Let $r = \max\{r_{i_1},\dots r_{i_k}\}$. Then,
    \begin{equation*}
        \Gamma_{i, j_1,\dots, j_k}^{i_1,\dots, i_k} = \left(H^{i_1}_{j_1}\right)^{r/r_{i_1}} \cdots \left(H^{i_k}_{j_k}\right)^{r/r_{i_k}} \left(P_i^{i_1,\dots, i_k}\right)^r,
    \end{equation*}
    is a conserved quantity and when $H^{i_1}_{j_1},\dots,H^{i_k}_{j_k}$ are all Casimirs of $\Gin^{i_1},\cdots,\Gin^{i_k}$, then the conserved quantity is a Casimir of $J$.

    \item Suppose $H_i(\xout)$ is a non-Casimir conserved quantity of $\Gout$ containing $x_{i_1},\dots x_{i_k}$, where $i_1,\dots,i_k$ are a subset of the indices into which $(\Gin^1,\dots,\Gin^m)$ are embedded. Let $H^{i_1}_1,\dots,H^{i_k}_1$ be the Hamiltonians of the input graphs $\Gin^{i_1},\cdots,\Gin^{i_k}$. Then,
    \begin{equation*}
        Z_{i, i_1,\dots, i_k} = \sum_{l \in \mathcal{U}_i}\left(H^{i_1}_1, \dots, H^{i_k}_1, \xoutTilde \right),
    \end{equation*}
    is also a conserved quantity.
\end{enumerate}
This can be summarised as follows: First, replace the $x_1,\dots,x_m$ in the Casimir(s) of $\Gout$ with any combination of conserved quantities from the corresponding input graphs, being sure that each term in the resulting product is of the same polynomial order by raising the terms to an appropriate power. Then replace $x_1,\dots,x_m$ in any other non-Casimir conserved quantities of $\Gout$ with the Hamiltonians of $\Gin^1,\dots,\Gin^m$. This will produce the conserved quantities of $J$.

To illustrate this procedure, consider the embedding of two directed five-cycles into a third skip graph with six vertices, as shown in \cref{fig:FourInFour}.
\begin{figure}[htbp]
\centering
\includegraphics[width=0.5\textwidth]{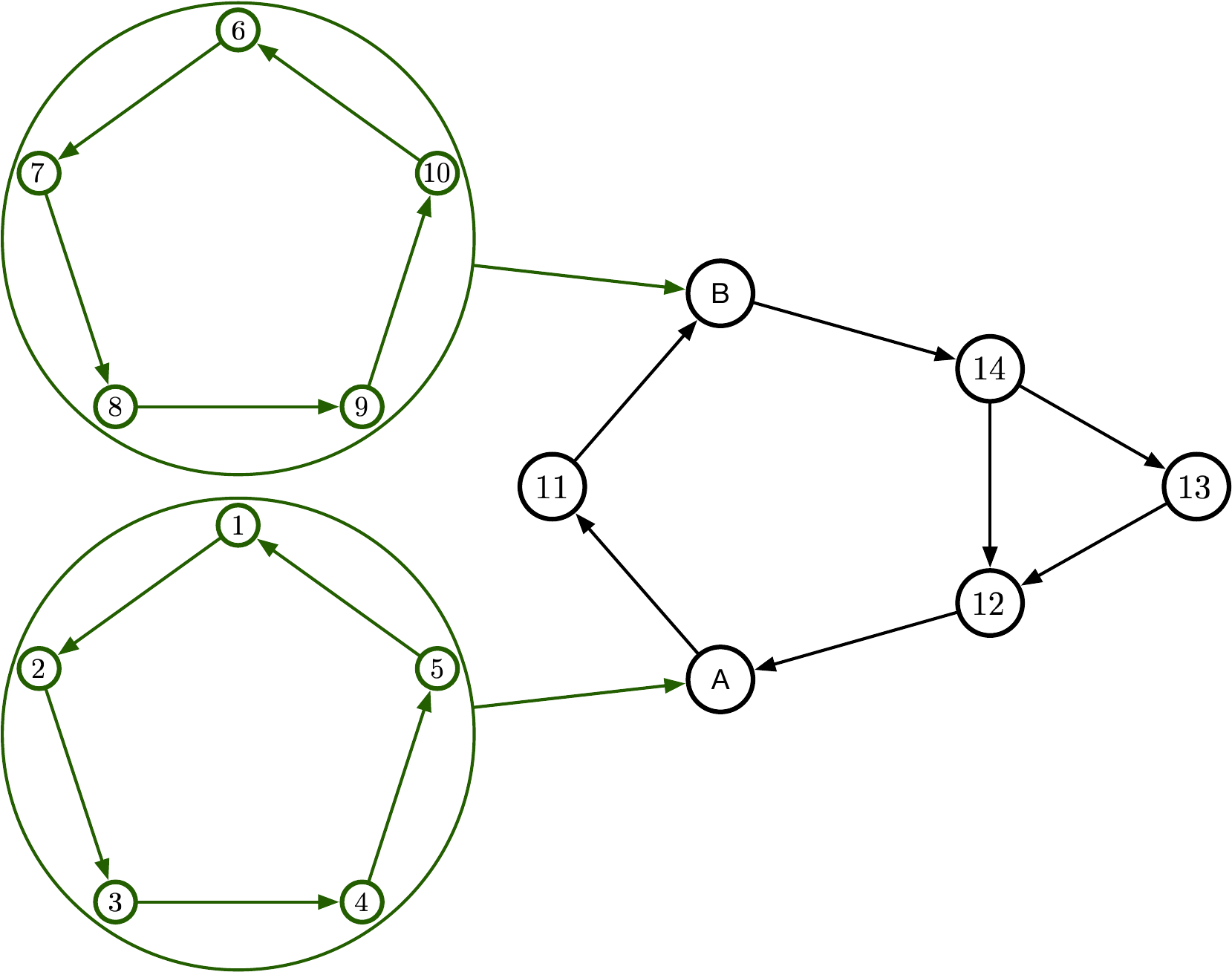}
\caption{When we embed two directed five-cycles into a skip-graph with six vertices, we require eight algebraically independent conserved quantities, which can be constructed using the conserved quantities of the constituent graphs.}
\label{fig:FourInFour}
\end{figure}
We have labelled the vertices for simplicity in constructing the conserved quantities. The conserved quantities for one of the five cycles are,
\begin{align*}
& F_1 = x_1 + x_2 + x_3 + x_4 + x_5\\
& F_2 = x_1 x_3+x_5 x_3+x_1 x_4+x_2 x_4+x_2 x_5\\
& F_3 = x_1 x_2 x_3 x_4 x_5.
\end{align*}
The conserved quantities for the skip-graph, in terms of their labelled vertices, are,
\begin{align*}
& H_1 = x_A+x_B+x_{11}+x_{12}+x_{13}+x_{14}\\
& H_2 = x_A x_B+x_Ax_{13} +x_Ax_{14} +x_Bx_{12} +x_Bx_{13} +
x_{11} x_{12}+x_{11} x_{13}+x_{11} x_{14} \\
& H_3 =  x_A x_Bx_{13}\\
& H_4 = x_A x_B x_{11} x_{12} x_{14} 
\end{align*}
We now organise the constructed conserved quantities into groups depending on which classes of inner-graph conserved quantities are substituted into the outer-graph conserved quantities.
\\
\noindent\textbf{Group 1: Casimirs into Casimir}\\
\hspace*{2em}
$(x_1 x_2 x_3 x_4 x_5)(x_6 x_7 x_8 x_9 x_{10}) x_{11}^5 x_{12}^5 x_{14}^5$\\
\hspace*{2em}
$(x_1 x_2 x_3 x_4 x_5) (x_6 x_7 x_8 x_9 x_{10}) x_{13}^5$\\

\noindent\textbf{Group 2: Mixed Casimir and Second Order Conserved Quantities in Casimir}\\
\hspace*{2em}
$x_1 x_2 x_3 x_4 x_5 \left(x_6 x_8+x_{10} x_8+x_6 x_9+x_7 x_9+x_7 x_{10}\right)^{5/2} x_{11}^5 x_{12}^5
   x_{14}^5$\\
\hspace*{2em}
$x_1 x_2 x_3 x_4 x_5 \left(x_6 x_8+x_{10} x_8+x_6 x_9+x_7 x_9+x_7 x_{10}\right)^{5/2} x_{13}^5$\\
\hspace*{2em}
$\left(x_1 x_3+x_5 x_3+x_1 x_4+x_2 x_4+x_2 x_5\right)^{5/2} x_6 x_7 x_8 x_9 x_{10} x_{11}^5 x_{12}^5
   x_{14}^5$\\
\hspace*{2em}
$\left(x_1 x_3+x_5 x_3+x_1 x_4+x_2 x_4+x_2 x_5\right)^{5/2} x_6 x_7 x_8 x_9 x_{10} x_{13}^5$\\

\noindent\textbf{Group 3: Second Order Conserved Quantities in Casimir}\\
\hspace*{2em}
$\left(x_1 x_3+x_5 x_3+x_1 x_4+x_2 x_4+x_2 x_5\right) \left(x_6 x_8+x_{10} x_8+x_6 x_9+x_7 x_9+x_7
   x_{10}\right) x_{11}^2 x_{12}^2 x_{14}^2$ \\
\hspace*{2em}
$\left(x_1 x_3+x_5 x_3+x_1 x_4+x_2 x_4+x_2 x_5\right) \left(x_6 x_8+x_{10} x_8+x_6 x_9+x_7 x_9+x_7
   x_{10}\right) x_{13}^2$\\

\noindent\textbf{Group 4: Mixed Hamiltonians and Second Order Conserved Quantities in Casimir}\\
\hspace*{2em}
$\left(x_1+x_2+x_3+x_4+x_5\right)^2 \left(x_6 x_8+x_{10} x_8+x_6 x_9+x_7 x_9+x_7 x_{10}\right) x_{11}^2 x_{12}^2 x_{14}^2$ \\
\hspace*{2em}
$\left(x_1+x_2+x_3+x_4+x_5\right)^2 \left(x_6 x_8+x_{10} x_8+x_6 x_9+x_7 x_9+x_7 x_{10}\right) x_{13}^2$\\
\hspace*{2em}
$\left(x_1 x_3+x_5 x_3+x_1 x_4+x_2 x_4+x_2 x_5\right) \left(x_6+x_7+x_8+x_9+x_{10}\right)^2 x_{11}^2
   x_{12}^2 x_{14}^2$\\
\hspace*{2em}
$\left(x_1 x_3+x_5 x_3+x_1 x_4+x_2 x_4+x_2 x_5\right) \left(x_6+x_7+x_8+x_9+x_{10}\right)^2 x_{13}^2$\\

\noindent\textbf{Group 5: Hamiltonians in Casimir}\\
\hspace*{2em}
$\left(x_1+x_2+x_3+x_4+x_5\right) \left(x_6+x_7+x_8+x_9+x_{10}\right) x_{11} x_{12} x_{14}$ \\
\hspace*{2em}
$\left(x_1+x_2+x_3+x_4+x_5\right) \left(x_6+x_7+x_8+x_9+x_{10}\right) x_{13}$\\

\noindent\textbf{Group 6: Hamiltonians in Second Order Conserved Quantity}\\
\begin{multline*}
    \left(x_1+x_2+x_3+x_4+x_5\right) \left(x_6+x_7+x_8+x_9+x_{10}\right)+\\
    x_{12}
   \left(x_6+x_7+x_8+x_9+x_{10}\right)+
   x_{13} \left(x_6+x_7+x_8+x_9+x_{10}\right)+\\
   \left(x_1+x_2+x_3+x_4+x_5\right) x_{13}+\left(x_1+x_2+x_3+x_4+x_5\right)
   x_{14}+\\ x_{11}
   x_{12}+  x_{11} x_{13} + x_{11} x_{14}
\end{multline*}
\\
\noindent\textbf{Group 7: Hamiltonians in Hamiltonian}\\
\hspace*{2em}
$x_1+x_2+x_3+x_4+x_5+x_6+x_7+x_8+x_9+x_{10}+x_{11}+x_{12}+x_{13}+x_{14}$\\

Straightforward computation shows that these are conserved quantities for the graph resulting from the embedding. These conserved quantities can be re-organised into polynomials of order $1$, $2$, $3$, $5$, $6$, $10$, $15$ and $25$. Analysing the matrix corresponding to the graph constructed from the embedding shows that eight conserved quantities are required for the graph to be integrable. The eight distinct polynomials thus show that this graph produces integrable dynamics as expected. 

The proof of \cref{thm:Embed} can be readily extended to the case when multiple inner graphs with Itoh-style conserved quantities are embedded into an outer graph with Itoh-style conserved quantities, when we use the conserved quantity construction procedure just outlined. In this case, the matrix for the resulting graph is simply decomposed into additional blocks and the same logic is applied, but in a combinatorial way. Moreover, we note that the proof structure relies only on the fact that the \textit{outer} graph admits Itoh-style conserved quantities. Meanwhile, the inner graph(s) can be far more generic, having themselves been constructed from an embedding procedure. An example of this is given in \ref{sec:DoubleEmbedding}. This leads to the following (recursive) definition and theorem.

\begin{definition}[Embedding Graph] A graph $J$ is an embedding graph if:
\begin{enumerate}
    \item There are graphs $\Gin^1,\dots,\Gin^m$ and $\Gout$ that all admit Itoh-style conserved quantities and are integrable in the sense of \cref{def:LAIntegrability} using those conserved quantities, and $J = (\Gin^1,\dots,\Gin^m) \hookrightarrow_{(v_1,\dots,v_m)} \Gout$.
    \item There are embedding graphs $\Gin^1,\dots,\Gin^m$ and a graph $\Gout$ that admits Itoh-style conserved quantities and is integrable in the sense of \cref{def:LAIntegrability} using those conserved quantities, and $J = (\Gin^1,\dots,\Gin^m) \hookrightarrow_{(v_1,\dots,v_m)} \Gout$.
\end{enumerate}
\label{def:EmbeddingGraph}
\end{definition}

Our final result follows from an appropriate (but combinatorial) modification to the proof of \cref{thm:Embed} to account for the generalised conserved quantity construction procedure and the observation that the matrix corresponding to the generated graph has a simple block decomposition. Notice, the integrability of the component graphs provides the guarantee that the constructed conserved quantities are independent, without needing to check the maximal Jacobian minors.

\begin{theorem} If $J$ is an embedding graph, then $J$ is integrable. \hfill\qed
\label{thm:GeneralEmbed}    
\end{theorem}

\section{Taxonomy of Integrable Graphs (Up to 6 Vertices)}\label{sec:Taxonomy}
We use the theoretical results developed in the previous sections to organise the dynamics generated by the 21,419 connected oriented directed graphs with up to 6 vertices. All computations were performed in Mathematica, version 14. Except for one figure, all graphs shown in this section were automatically laid out with Mathematica. Of those 21,419 graphs, we find that 57 are integrable. All these graphs admit a sufficient number of Itoh-style conserved quantities that commute under the quadratic bracket. We suspect that 184 graphs produce chaotic behaviour, as a result of numerical evaluation. While they admit a numer of Itoh-style conserved quantities, it is not sufficiently many to satisfy the requirements for Liouville-Arnold integrability. We do not (and cannot yet) prove these graphs have no other non-polynomial conserved quantities, but numerical chaos testing suggests it. The remaining 21,178 graphs asymptotically decay to a lower-order system and do not have an interior fixed point in a higher-order unit simplex.  Our results are summarised in \cref{tab:table_one}.
\begin{table}[htp!]\small
    \centering
    \begin{tabular}{ccccc}
        Vertex Count & Known Integrable & Suspected Chaotic & No Interior Fixed Point & Total\\
        \hline
        3 & 1 & 0 & 1 & 2\\
        4 & 3 & 0 & 31 & 34\\
        5 & 11 & 4 & 520 & 535\\
        6 & 42 & 180 & 20,626 & 20,848\\
        \hline
        \textbf{Totals} & 57 & 184 & 21,178 & 21,419
    \end{tabular}
    \caption{Breakdown by vertex size of all graphs that are integrable}
    \label{tab:table_one}
\end{table}

To derive these statistics, we used Brendan McKay’s database of graphs \cite{M24}. We first identified those graphs that were strongly connected; i.e., those graphs for which there is a directed path between any pair of vertices \cite{G23a}. The dynamics of these graphs necessarily asymptotically decay to the dynamics of a lower-order system, as shown by Paik and Griffin \cite{PG23}. This reduced the number of potentially integrable graphs down to 4,313.

To sort through these remaining graphs, we used an algorithm for analyzing each graph:
\begin{enumerate}
    \item Numerically check if the Dynamical System produced by the graph has an interior fixed point.
    \item Look for Itoh-style conserved quantities through the Casimirs of the graph, the graph structure, and Mathematica calculation.
    \item If there were not sufficient conserved quantities, we ran the Wolfe Algorithm \cite{WSSV85,S96} to confirm the numerical existence of the maximum Lyapunov exponent greater than zero.
    \item If there were sufficient conserved quantities, we (manually) categorised the graph in the taxonomy.
\end{enumerate}

For graphs with no interior fixed point, the dynamics asymptotically approach a face of the unit simplex. In this case,  trajectories may assymptotically approach an $n$-torus in a foliation of the phase space restricted to a face of the simplex making up the (higher dimensional) phase space. Consequently, there dynamics cannot be Liouville-Arnold integrable because the trajectories cannot foliate the original phase space. They may (in some sense) asymptotically foliate a lower-dimensional phase space or they may be integrable (solvable) in a sense that is distinct from \cref{def:LAIntegrability}. We are not aware of such a result in the literature.

We numerically identified the remaining graphs with an interior fixed point. It is possible, but not computationally efficient, to use Mathematica's symbolic solver to identify these graphs. Instead, we numerically integrated the system of differential equations generated from each graph to obtain numerical flows $x_i(t)$, where $i$ ranged over the number of vertices in the graph in question, using an initial condition $\mathbf{x} \in \mathrm{int}(\Delta_{n-1})$ for appropriate $n$. We then numerically computed, 
\begin{equation*}
    \bar{x}_i = \frac{1}{T_f - T_0}\int_{T_0}^{T_f} x_i(t) \,dt,
\end{equation*}
where $T_0 \gg 0$. Since the initial condition is in the interior of an appropriate unit simplex, we know that $\bar{x}_i > 0$ for all $i$ if the dynamics admit an interior fixed point. If there is some $i$ for which $\bar{x}_i < \tau$, where $\tau \ll 1$ is a threshold, we then conclude the dynamics admit no interior fixed point, and they collapse to the boundary of $\Delta_{n-1}$. In that case, the trajectories asymptotically approach the orbits of a lower-dimensional replicator equation. For our numerical experiments, we used $\tau = 10^{-6}$ and via Mathematica's \texttt{Chop} function and we set $T_0 = 250$.

Then, we attempted to find a sufficient number of Itoh-style conserved quantities for each remaining graph. For the graphs that did not have enough Itoh-Style conserved quantities, we used the Wolf Algorithm \cite{WSSV85,S96} find numerical evidence of chaotic behaviour. Many of the graphs exhibit weak chaos with very small (but positive) Lyapunov exponents. Consequently, we ran the Wolf algorithm over several initial conditions for each graph to identify possible signs of chaotic behaviour in the dynamics. In this way, all 21,419 were categorised (see \cref{tab:table_one}).

While executing the Wolf Algorithm, we found that several 5 vertex graphs exhibited signs of numerical chaos, suggesting that these are the smallest possible graphs to produce chaotic dynamics. One of the simplest of these graphs is shown in \cref{fig:bowtie} and appears to be an example of a coupled oscillator in the replicator dynamics, as it is composed of two directed three cycles joined at a single vertex. 
\begin{figure}[htp!]
\begin{tblr}{
  colspec = {Q[c, m]Q[c, h]},
  stretch = 0
}
\includegraphics[width=0.45\textwidth]{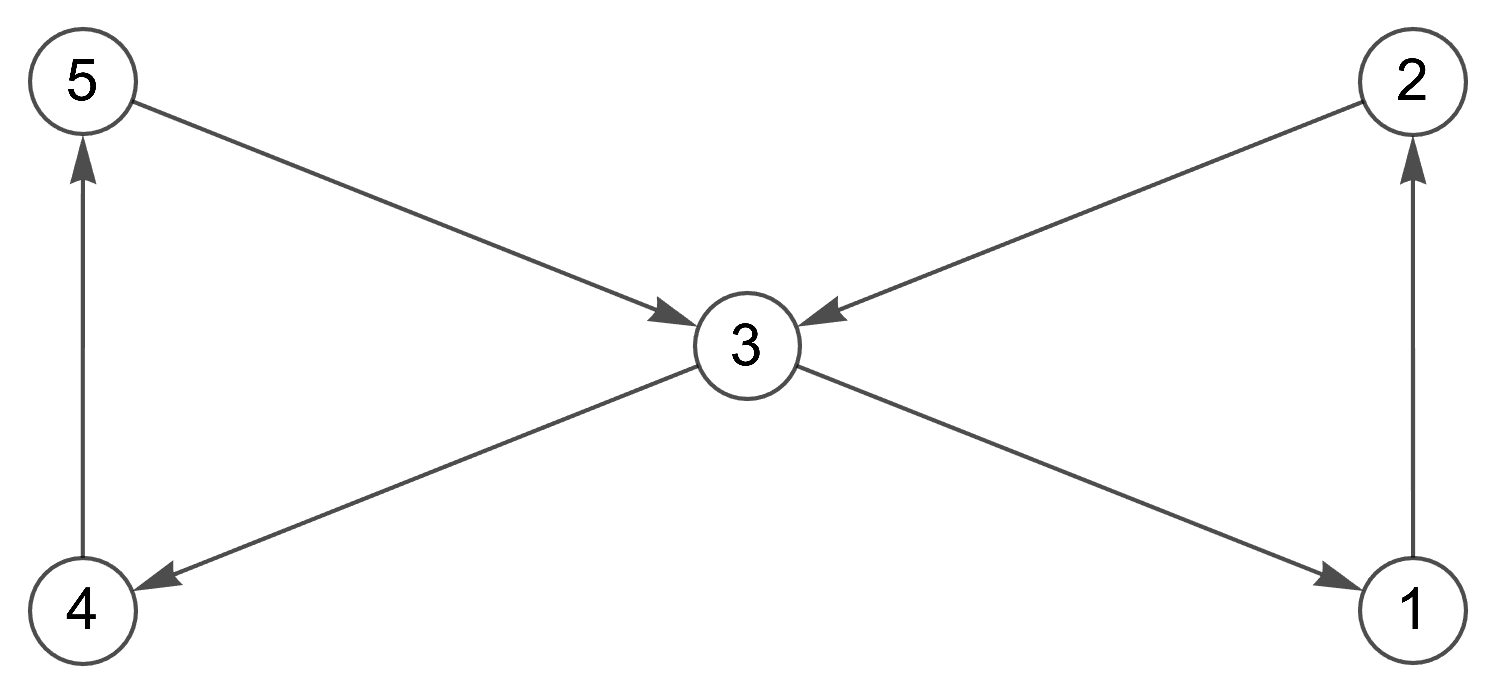}&
\includegraphics[width=0.45\textwidth]{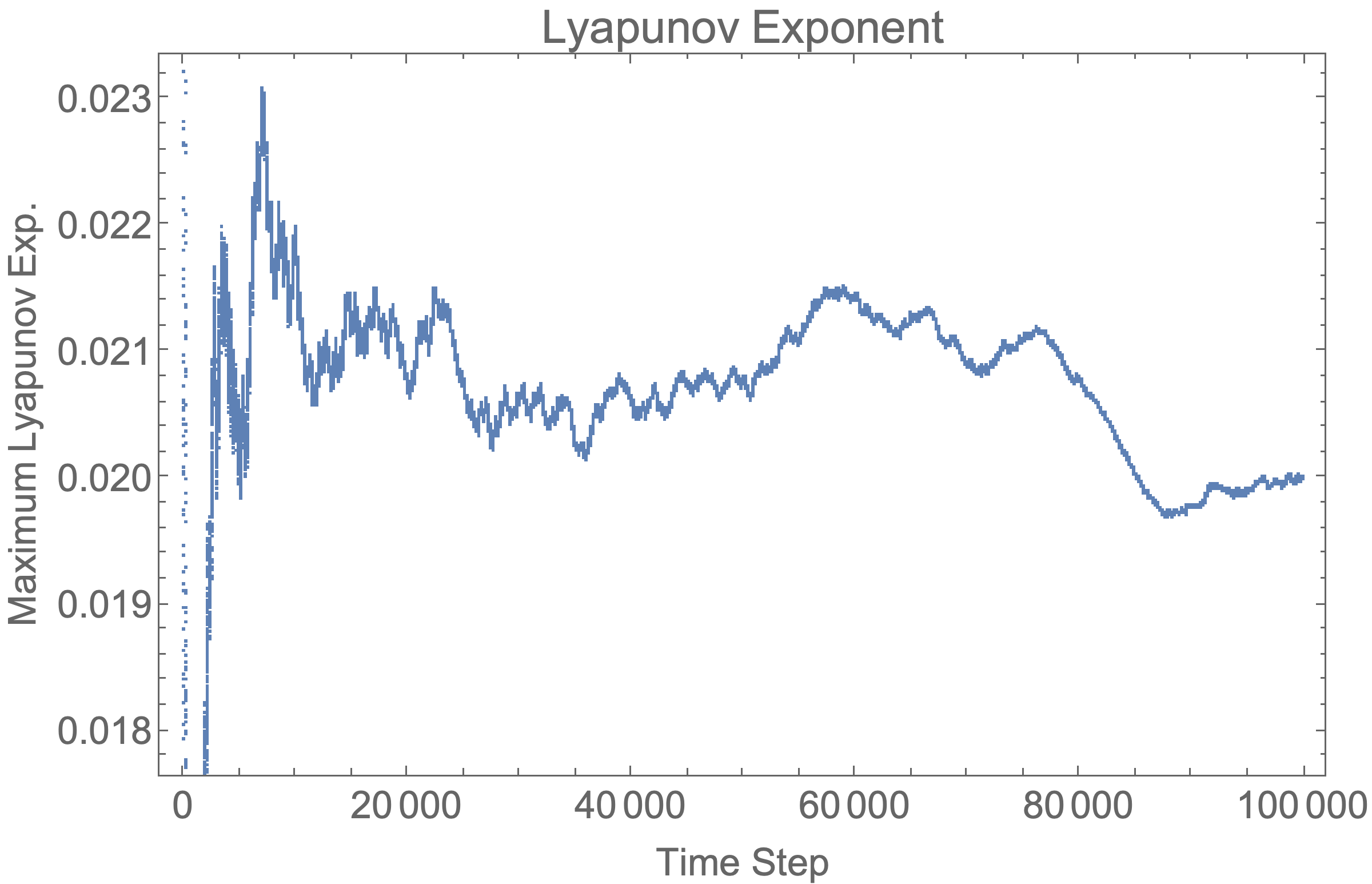}
\end{tblr}
\caption{(Left) A simple graph that produces chaotic dynamics appears to be a coupled oscillator. (Right) Output of the Wolf Algorithm indicating that the system has positive Lyapunov exponent.}
\label{fig:bowtie}
\end{figure}
While we have not formally proved that this graph admits dynamics with positive Lyapunov exponent, its small size may make it amenable to such an analysis in future work.

It is worth noting that each graph that does not have Itoh-style conserved quantities may have other conserved quantities. However, every graph that lacked Itoh-style conserved quantities displayed evidence of positive Lyapunov exponent from the Wolf Algorithm. Therefore, either each graph has more conserved quantities and the Wolf Algorithm is a flawed method of analysis for these systems or no more conserved quantities exist.

In our investigation, we have found that, with limited exceptions, integrable graphs up to 6 vertices are characterised by a modification to a directed cycle (i.e., an element of the Volterra lattice). We note that this modification need not preserve symmetry, as shown by the skip-vertex graphs. Based on work presented in this paper as well as the previous work done by Bogoyavlenskij  \cite{BIY08} and Evripidou  et al.\cite{EKV22}, we can categorise all but three six vertex graphs into four families: Bogoyavlenskij  graphs (denoted $B(n,k)$), Skip-vertex graphs, Cloned graphs and Embedding graphs. The remaining graphs are considered Holdout graphs. Raw counts of each family are shown in \cref{tab:graphfamilies}.
\begin{table}[htp!]
    \centering
    \begin{tabular}{lccccc}
    &\multicolumn{4}{c}{Vertex size} & \\ 
    \hline
        Family & 3 & 4 & 5 & 6 & Total\\
        \hline
        $B(n, k)$ & 1 & 1 & 2 & 2 & 6\\
        Skip-Vertex & 0 & 1 & 0 & 3 & 4\\
        Cloned & 0 & 1 & 7 & 24 & 32\\
        Embedding & 0 & 0 & 1 & 11 & 12\\
        Holdout & 0 & 0 & 0 & 3 & 3
    \end{tabular}
    \caption{Breakdown of Integrability by Family and Vertex Size}
    \label{tab:graphfamilies}
\end{table}
We provide illustrative examples from each family in \ref{sec:Examples}. 

\subsection{Holdout Graphs}

There are three integrable graphs that do not fall into any of these families. Of those three, two appear similar. They  both are cycles (a four cycle and a five cycle, respectively) with an extra vertex with equal in and out degrees of two. The graphs and their corresponding conserved quantities are shown in \cref{fig:2 Holdouts}. Notice these all have Itoh-style conserved quantities made of non-edges and cycles, as expected. 
\begin{figure}[htp!]
\begin{tblr}{
  colspec = {Q[c, m]Q[c, h]},
  stretch = 0
}
\includegraphics[width=0.4\textwidth]{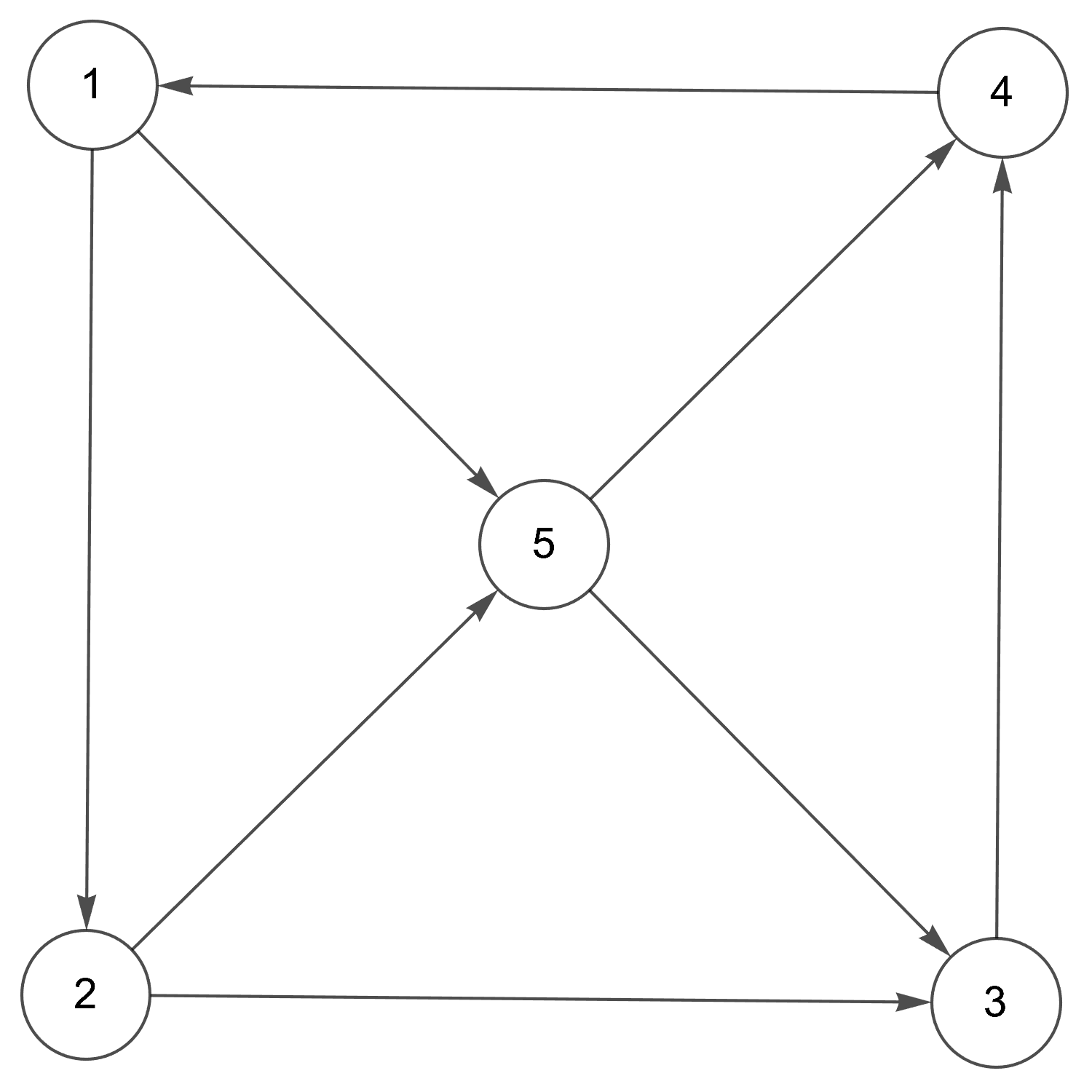} & 
\includegraphics[width=0.45\textwidth]{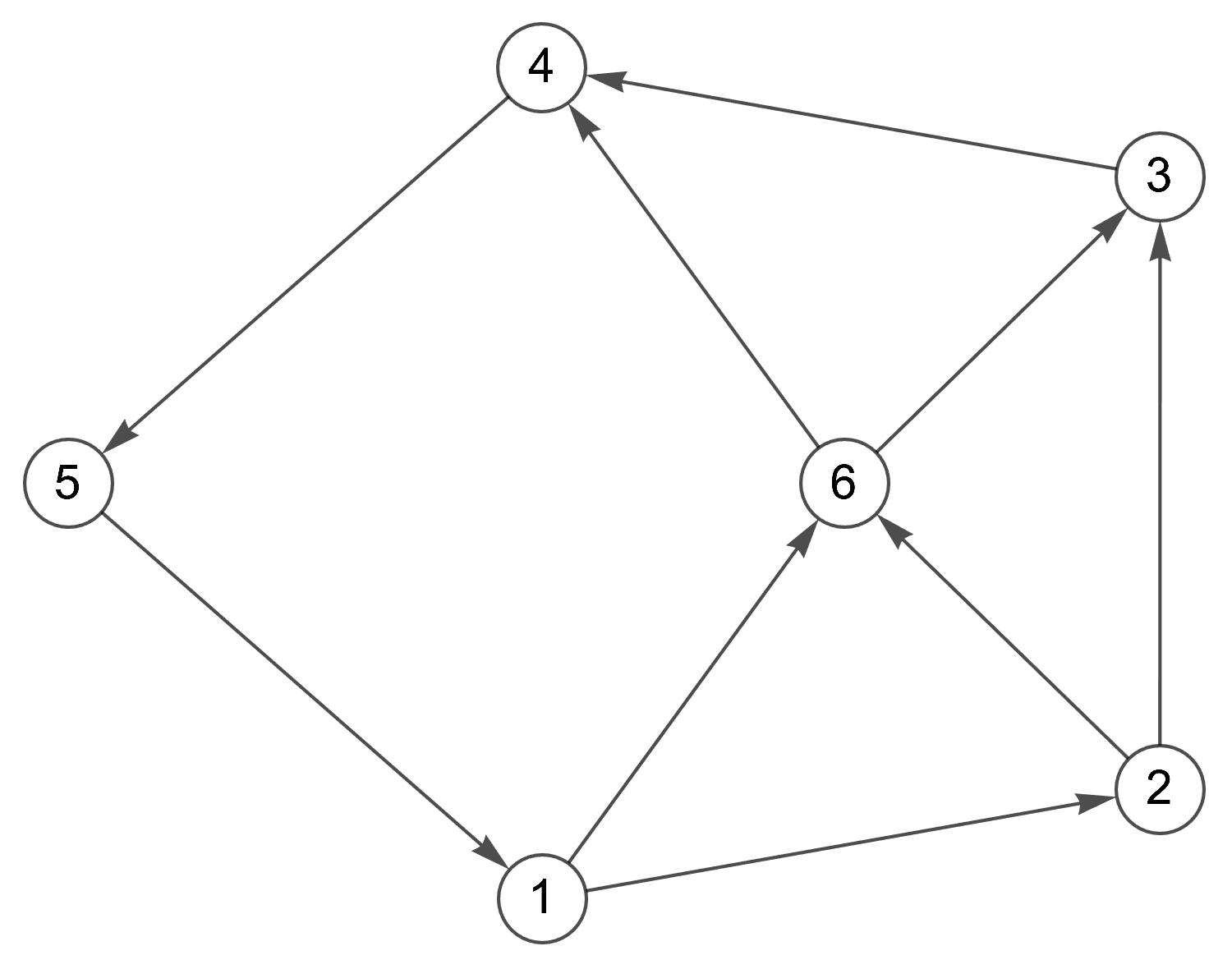}\\
\begin{minipage}{0.45\textwidth}
\begin{align*}
&H_1 = x_1+x_2+x_3+x_4+x_5\\
&H_2 = x_1x_3 + x_2x_4\\
&H_3 = x_1x_4x_5\\
&H_4 = x_1x_2x_3x_4
\end{align*}
\end{minipage}& 
\begin{minipage}{0.45\textwidth}
\begin{align*}
&H_1 = x_1+x_2+x_3+x_4+x_5+x_6\\
&\begin{aligned}
&H_2 = x_1x_3+x_1x_4+x_2x_4+\\
&\hspace*{4em}x_2x_5+x_3x_5+x_5x_6\\
\end{aligned}\\
&H_3 = x_1x_4x_5x_6\\
&H_4 = x_1x_2x_3x_4x_5
\end{align*}
\end{minipage}
\end{tblr}
\caption{2 Holdout graphs that appear to be related}
\label{fig:2 Holdouts}
\end{figure}

Further analysis suggests these graphs are part of a large family containing $n+1$ vertices with two in-edges leading from Vertices $i$ and $i+1$ to Vertex $n+1$ and two out-edges leading to Vertex $n+1$ to Vertices $i+2$ and $i+3$. The next two graphs in this family are shown in \cref{fig:TwoSpoke}. 
\begin{figure}[htbp]
\centering
\begin{tblr}{
  colspec = {Q[c, m]Q[c, h]},
  stretch = 0
}
\includegraphics[width=0.45\textwidth]{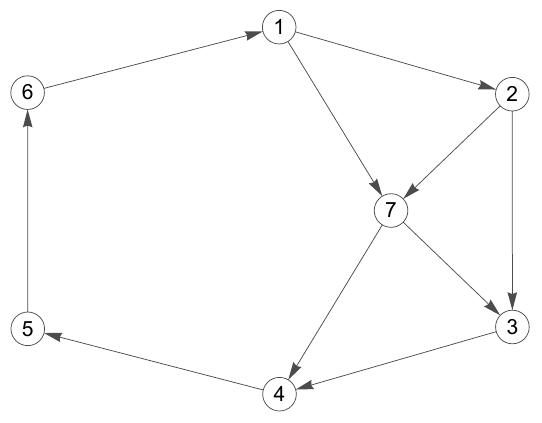} & 
\includegraphics[width=0.45\textwidth]{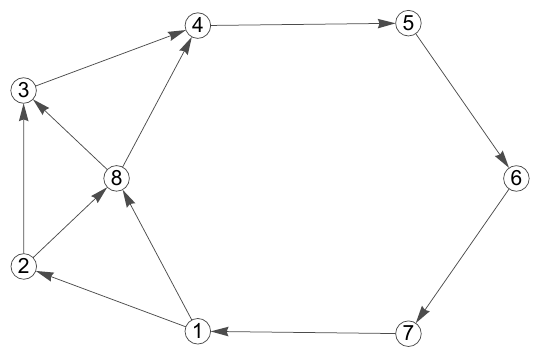}
\end{tblr}
\caption{The next two elements of a ``two-spoke'' family are integrable, suggesting a new infinite family.}
\label{fig:TwoSpoke}
\end{figure}
These graphs are integrable, having a sufficient number of commuting (Itoh-style) conserved quantities that commute. (See the SI for details.) It is left as future work whether these graphs are part of an even larger family related to wheel graphs, or if this characterisation is complete. We leave this analysis for future research. It is worth noting that these graphs also can be described as modifications of directed cycles.  

One final integrable graph with six vertices stands alone. It can be described as a balanced tournament with an extra vertex. The balanced tournament is composed of Vertices 1-5. This is not the archetypal balanced tournament considered by Paik and Griffin \cite{PG23} or that emerges from the Bogoyavlenskij  construction\footnote{The archetypical balanced tournament on five vertices describes rock-paper-scissors-Spock-Lizard, in that order.}. Instead, it is isomorphic to this archetype. Vertex 6 has the opposite neighbourhood to Vertex 1. The graph and its conserved quantities are shown in \cref{fig:Headscratcher}.
\begin{figure}
    \centering
    \includegraphics[width=0.5\linewidth]{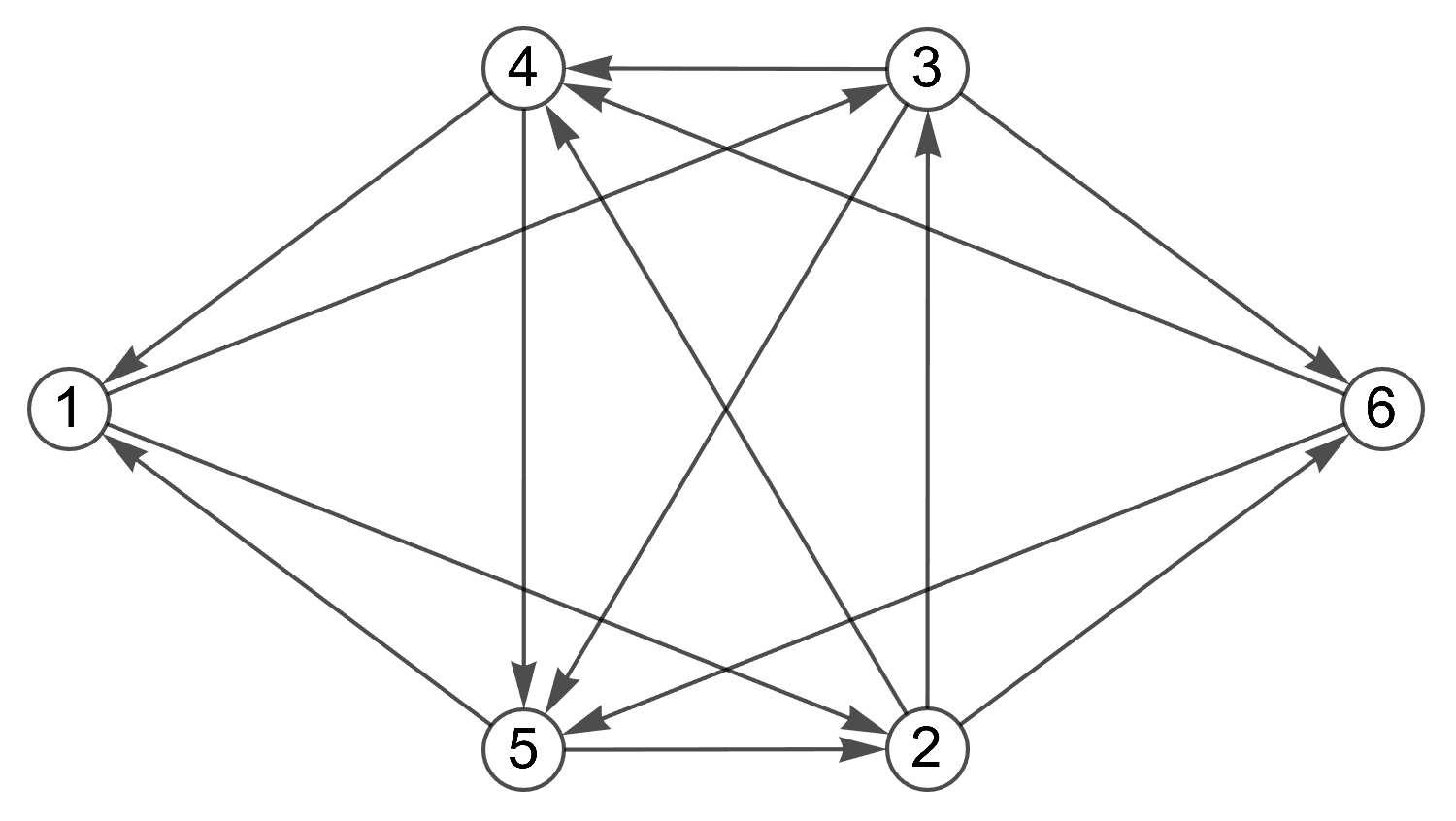}\\
    \begin{align*}
        &H_1 = x_1 + x_2 + x_3 + x_4 + x_5 + x_6\\
        &H_2 = x_1x_2x_4+x_1x_3x_4+x_1x_3x_5+x_2x_3x_5+x_2x_4x_5+x_2x_5x_6\\
        &H_3 = x_1x_6\\
        &H_4 = x_1x_2x_3x_4x_5
    \end{align*}
    \caption{An unclassified hold-out graph exhibiting a form of chiral symmetry.}
    \label{fig:Headscratcher}
\end{figure}
This graph suggests the potential for another operation, `anti-cloning', in which a vertex is introduced with a mirror image neighbourhood, imposing a kind of chiral symmetry into the graph. However, when this operation is performed  on the 5-cycle, the resulting graph is in the family of graphs suspected of chaotic behaviour by numerical analysis. Thus, this final graph indicates that, while we have been able to categorise the vast majority of the integrable graphs, there are still missing pieces to this puzzle to consider in future work.  

\section{Discussion and Future Directions}\label{sec:Conclusion}
In this paper, we generalised the concept of embeddings first introduced by Paik and Griffin \cite{PG23}. While Paik and Griffin showed that embeddings of tournaments were integrable, we extended this concept to show that any (set of) integrable graphs with Itoh-style conserved quantities can be embedded into another integrable graph with Itoh-style conserved quantities to produce an integrable graph. We also introduced a new family of integrable graphs (the skip-graph family) and used these graphs along with our results on graph embeddings and the prior work by Evripidou  et al. \cite{EKV22} to completely characterise the behaviour of the dynamics generated by all 21,419 oriented directed graphs with six or fewer vertices. In particular, we found 57 distinct integrable graphs, of which 54 fit into a taxonomy of four categories. Three hold-out graphs suggest the existence of both new integrable families as well as graph operations that preserve integrability. 

In addition to the integrable dynamics we have discussed, our analysis of graph structures and their resulting dynamics shows numerical evidence for Hamiltonian chaos (including weak chaos). In some cases (e.g., the graph shown in \cref{fig:bowtie}), the dynamics are simple enough that they may admit a formal proof for the existence of positive Lyapunov exponents. Of note, this chaotic behaviour is a function of the graph structure itself, a property also observed by Griffin et al. in networked replicator dynamics \cite{GSB22}. For some graphs, simple addition (or subtraction) of an edge will dramatically alter the underlying dynamics. For example, the graph shown in \cref{fig:ToFromChaos} produces chaotic behaviour, but removing the edge from Vertex 2 to Vertex 5 recovers the directed five cycle, an integrable graph. 
\begin{figure}[htbp]
\centering
\begin{tblr}{
  colspec = {Q[c, m]Q[c, h]},
  stretch = 0
}
\includegraphics[width=0.45\textwidth]{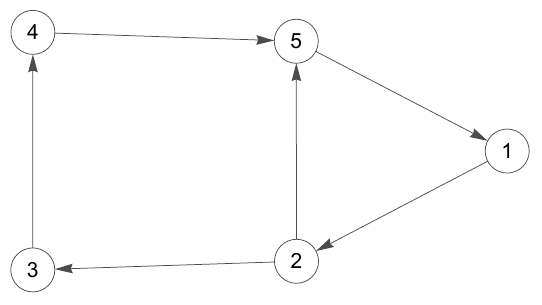}&
\includegraphics[width=0.45\textwidth]{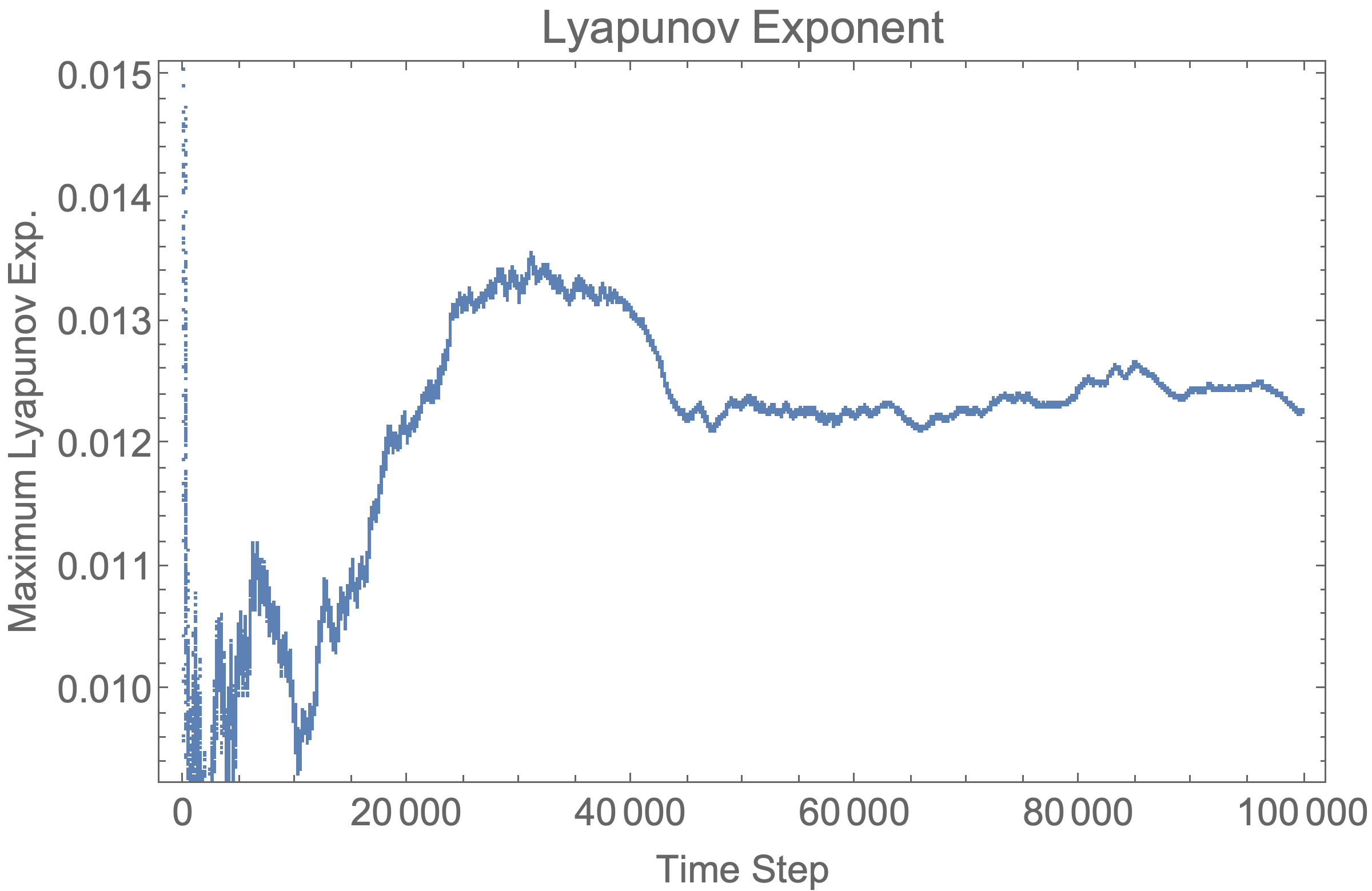}
\end{tblr}
\caption{A graph that produces chaotic behaviour, but that will produce integrable (quasi-periodic) behaviour if a single edge is removed.}
\label{fig:ToFromChaos}
\end{figure}
This suggests the dynamics arising from graphs may be ideal toy models for experimenting with the effects of the emergence and destruction of feedback loops in complex systems. Such feedback loops are known to occur in both biological and environmental dynamics, making the models and results presented in this paper potentially relevant in understanding complex natural systems.

In summary, this work represents a continuation of work begun by Moser, Kac and van Moerbeke on the Volterra lattice and extended by Itoh, Bogoyavlenskij  and Evripidou, Kassotakis and Vanhaecke. However, when taken as a whole, it suggests a program of research far more extensive in scope. There is a structure to the graphs that generate integrable replicator (Lotka-Volterra) dynamics, with families of ``seed graphs'' being used in cloning or embedding operations to generate sets of integrable graphs. It is unclear how many of these seed graph families there are, and whether they are all related in some way to the directed cycle, as we suspect. We note that the smallest integrable skip-graph has $4$ vertices; therefore this family starts with four vertices. It is an open question whether there are other integrable families of graphs whose smallest member is larger than four vertices and that is not constructed by embedding or cloning. Moreover, it is unclear whether there are other operations beyond cloning and embedding that generate new integrable graphs from integrable graphs. The final hold-out graph in our study suggests either a new family or a new operation that generates integrable graphs from other integrable graphs. This line of research is similar to the search for structure in the simple groups, and has the potential to quantify deep connections between combinatorial structures and integrable dynamics.

\section*{Acknowledgements}
M.V. was supported in part by an Erickson Discovery Grant from the Pennsylvania State University. C.G. was supported in part by the National Science Foundation under grant CMMI-1932991. M.V. and C.G. both thank the anonymous reviewers for their feedback.

\section*{Data and Code Availability}
Mathematica notebooks are provided as supplementary materials and contain the code needed to reproduce the results in this paper.

\appendix

\section{A Double Embedding}\label{sec:DoubleEmbedding}
By way of example, we illustrate the conserved quantities generated when two embedded graphs are themselves embedded in a graph with Itoh-style conserved quantities. The graph in question is shown in \cref{fig:DoubleEmbedding}.
\begin{figure}[htbp]
\centering
\includegraphics[width=0.5\textwidth]{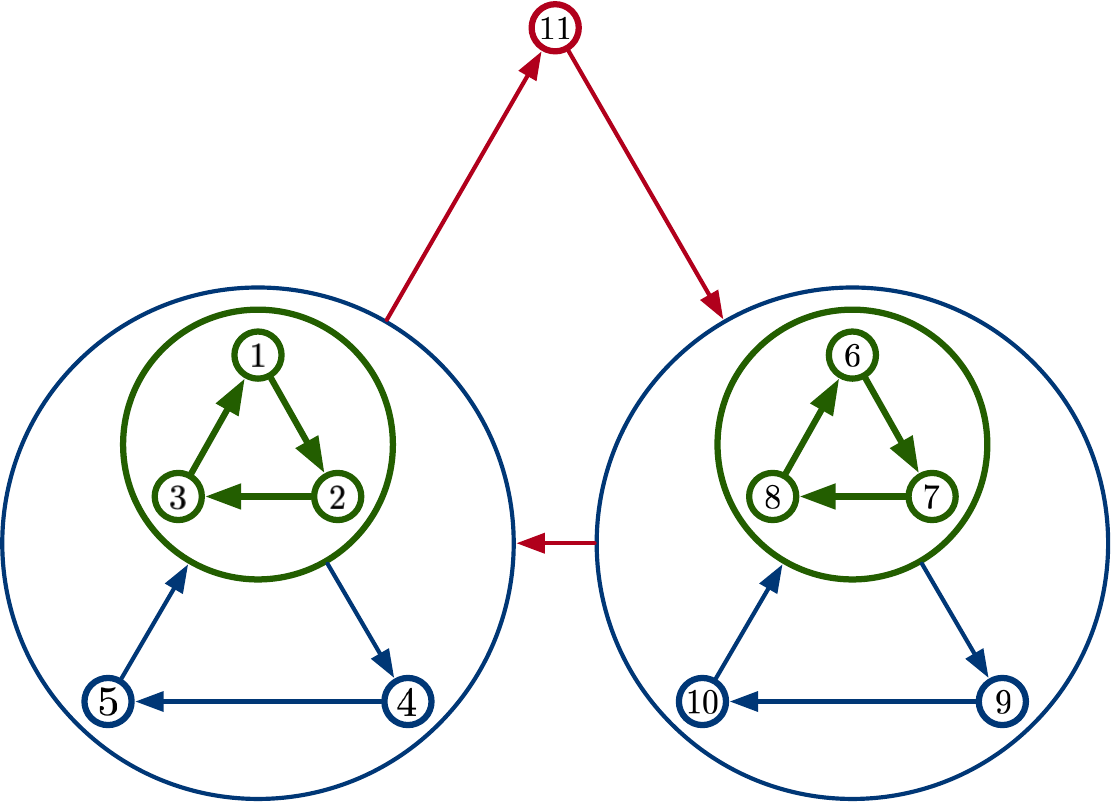}
\caption{Graphs can be recursively embedded within each other to create new integrable dynamics.}
\label{fig:DoubleEmbedding}
\end{figure}
As shown in \cite{PG23}, the conserved quantities for one of the embedded graphs are,
\begin{align}
    &H_1 = x_1 + x_2 + x_3 + x_4 + x_5\\
    &H_2 = (x_1 + x_2 + x_3)x_4x_5\\
    &H_3 = x_1x_2x_3x_4^3x_5^3.
\end{align}
These also follow immediately from \cref{thm:Embed}. Consequently, we can read off ten conserved quantities.
\begin{align*}
    &H_{10} = (x_1x_2x_3x_4^3x_5^3)(x_6x_7x_8x_9^3x_{10}^3)x_{11}^9\\
    &H_9 = \left[(x_1 + x_2 + x_3)x_4x_5\right]^3(x_6x_7x_8x_9^3x_{10}^3)x_{11}^9\\
    &H_8 = (x_1x_2x_3x_4^3x_5^3)\left[(x_6 + x_7 + x_8)x_9x_{10}\right]^3x_{11}^9\\
    &H_7 = (x_1 + x_2 + x_3 + x_4 + x_5)^9(x_6x_7x_8x_9^3x_{10}^3)x_{11}^9\\
    &H_6 = (x_1x_2x_3x_4^3x_5^3)(x_6 + x_7 + x_8 + x_9 + x_{10})^9x_{11}^9\\
    &H_5 = \left[(x_1 + x_2 + x_3)x_4x_5\right]\left[(x_6 + x_7 + x_8)x_9x_{10}\right]x_{11}^3\\
    &H_4 = (x_1 + x_2 + x_3 + x_4 + x_5)^3\left[(x_6 + x_7 + x_8)x_9x_{10}\right]x_{11}^3\\
    &H_3 = \left[(x_1 + x_2 + x_3)x_4x_5\right](x_6 + x_7 + x_8 + x_9 + x_{10})^3x_{11}^3\\
    &H_2 = (x_1 + x_2 + x_3 + x_4 + x_5)(x_6 + x_7 + x_8 + x_9 + x_{10})x_{11}\\
    &H_1 = (x_1 + x_2 + x_3 + x_4 + x_5) + (x_6 + x_7 + x_8 + x_9 + x_{10}) + x_{11}
\end{align*}
Computation by hand or using the computational tools included in the SI shows that these quantities commute under the bracket, as expected. 

\section{Selected Examples of Integrable graphs from the Defined Families}\label{sec:Examples}

\subsection{Bogoyavlenskij  Family}
The Bogoyavlenskij  family consists of all graphs $B(n,k)$, including the tournament graphs with an odd number of vertices. Technically, tournaments with $n$ even and $k = \tfrac{n}{2}$ are not Bogoyavlenskij  graphs and are not integrable, as their dynamics asymptotically converge to a lower order replicator dynamic. Tournaments are considered separately by Bogoyavlenskij  \cite{B88} and by Itoh \cite{I87}, who provides Itoh-style conserved quantities. In \cite{I87}, Itoh-style conserved quantities for all graphs in the Bogoyavlenskij  family are provided. The Bogoyavlenskij  family is characterised by their high degree of symmetry, as shown in \cref{fig:Bnk}. Notice the conserved quantities of these graphs are formed from combinations of cycles and non-cycles (non-edges). 

\begin{figure}[htbp]
\begin{tblr}{
  colspec = {Q[c, m]Q[c, h]},
  stretch = 0
}
\includegraphics[width=0.45\textwidth]{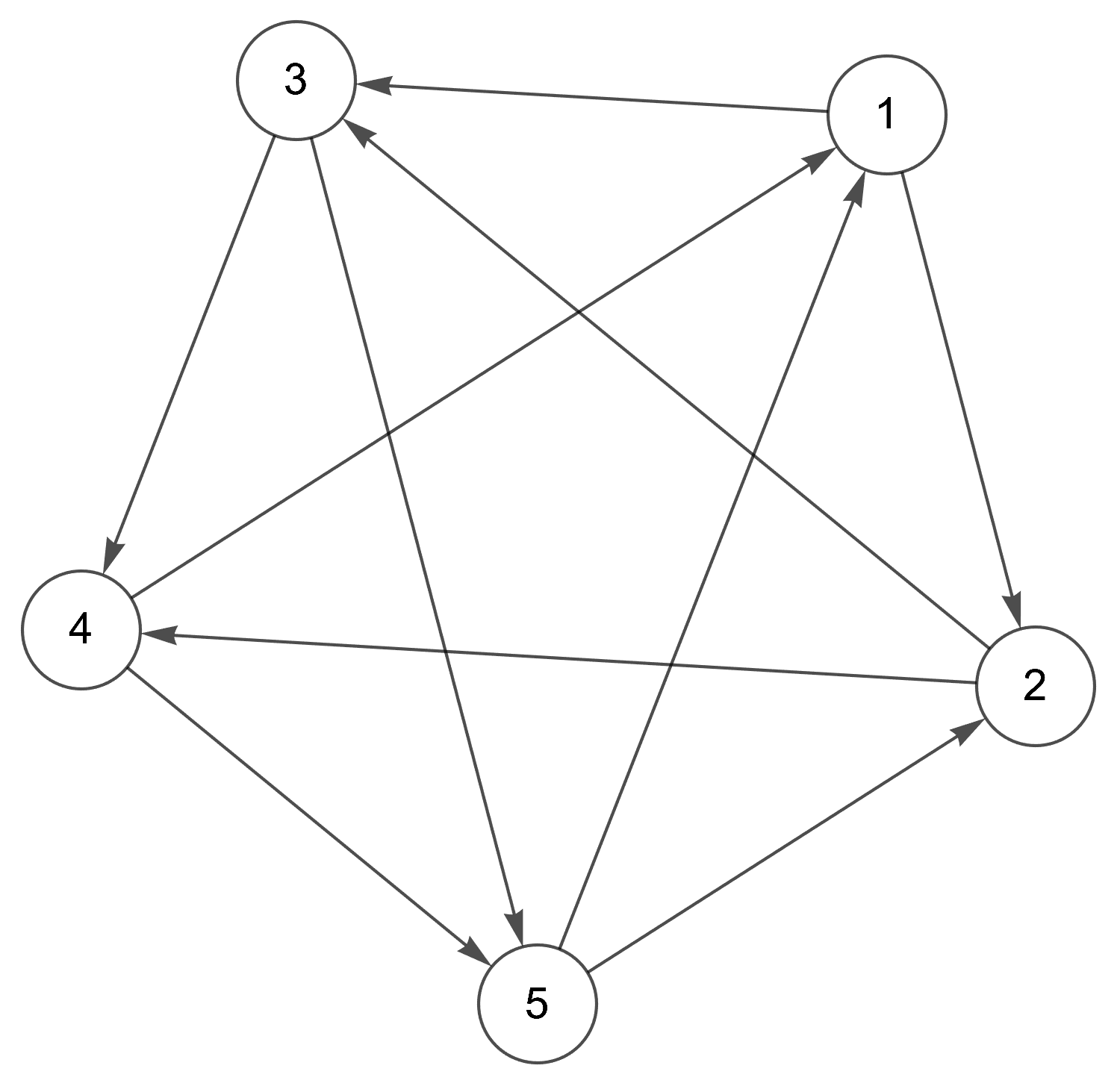} &  
\includegraphics[width=0.45\textwidth]{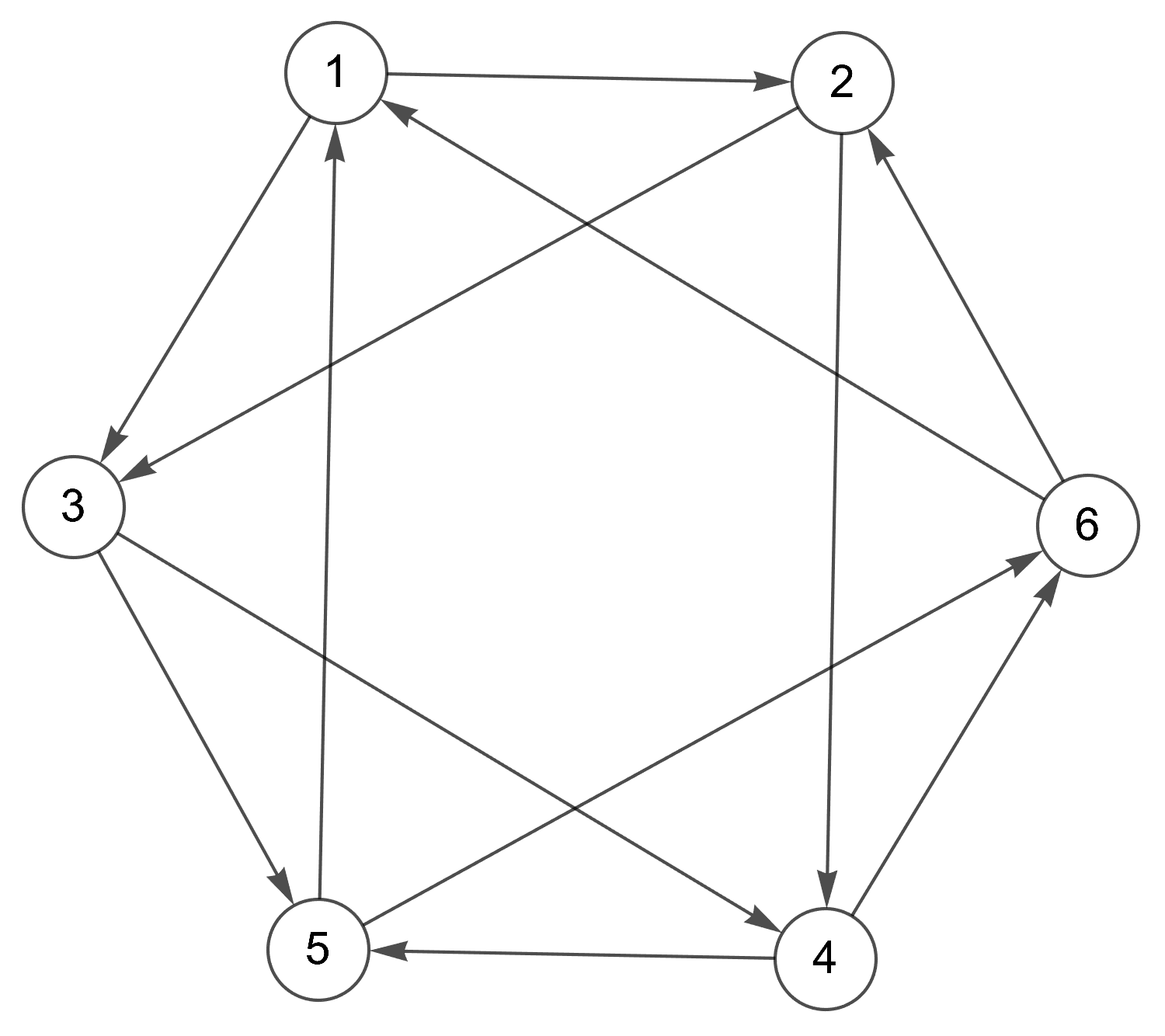}\\
\begin{minipage}{0.45\textwidth}\small\begin{align*}
    &H_1 = x_1 + x_2 + x_3 + x_4 + x_5\\
    &\begin{aligned}
    &H_2 = x_1x_2x_3 + x_2x_3x_4 + \\
    &\hspace*{4em}x_3 x_4 x_5 + x_1x_4x_5 + x_1x_2x_5\\
    \end{aligned}\\
    &H_3 = x_1x_2x_3x_4x_5
\end{align*}\end{minipage} & 
\begin{minipage}{0.45\textwidth}\small\begin{align*}
    &H_1 = x_1+x_2+x_3+x_4+x_5+x_6\\
    &H_2 = x_1 x_4+x_2 x_5+x_3 x_6\\
    &H_3 = x_1 x_3 x_5+x_2 x_4 x_6\\
    &H_4 = x_1 x_2 x_4 x_5+x_2 x_3 x_6 x_5+x_1 x_3 x_4 x_6\\
    &H_5 = x_1 x_2 x_3 x_4 x_5 x_6
\end{align*}
\end{minipage}
\end{tblr}
\caption{(Left) Balanced tournament graph with five vertices (species). (Right) Bogoyavlenskij  graph with the maximum number of edges on six vertices. The conserved quantities are shown below their respective graphs.}
\label{fig:Bnk}
\end{figure}

\subsection{Skip-Vertex Family}
We showed that all members of the skip-vertex family are integrable in \cref{thm:SkipVertex} and constructed their Itoh-style conserved quantities. There are 4 skip-vertex graphs with at most 6 vertices: one with 4 vertices, and three with 6 vertices. Two examples are shown  in \cref{fig:SkipVertex}
\begin{figure}[htbp]
\begin{tblr}{
  colspec = {Q[c, m]Q[c, h]},
  stretch = 0
}
\includegraphics[width=0.45\textwidth]{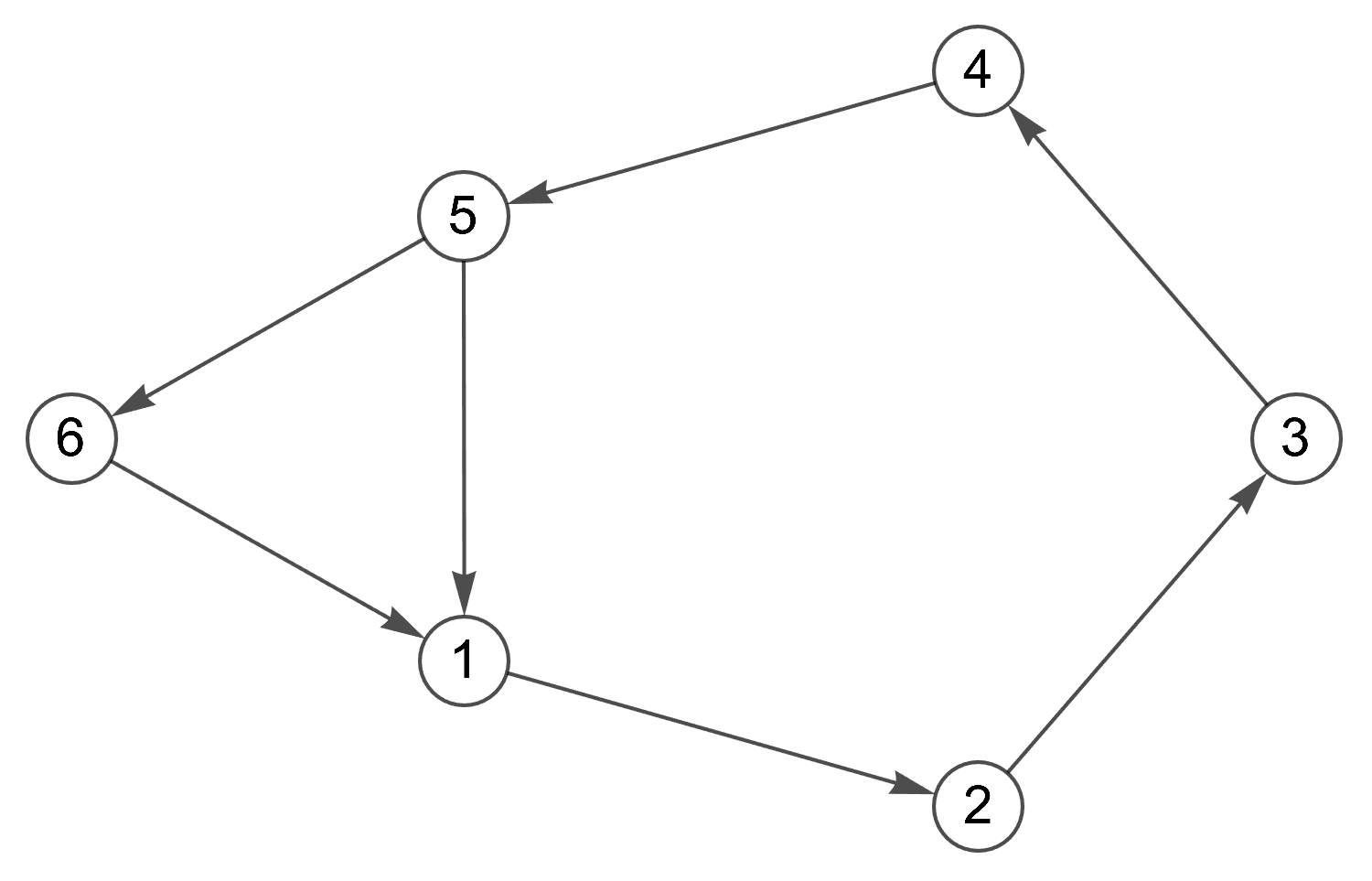} & 
\includegraphics[width=0.45\textwidth]{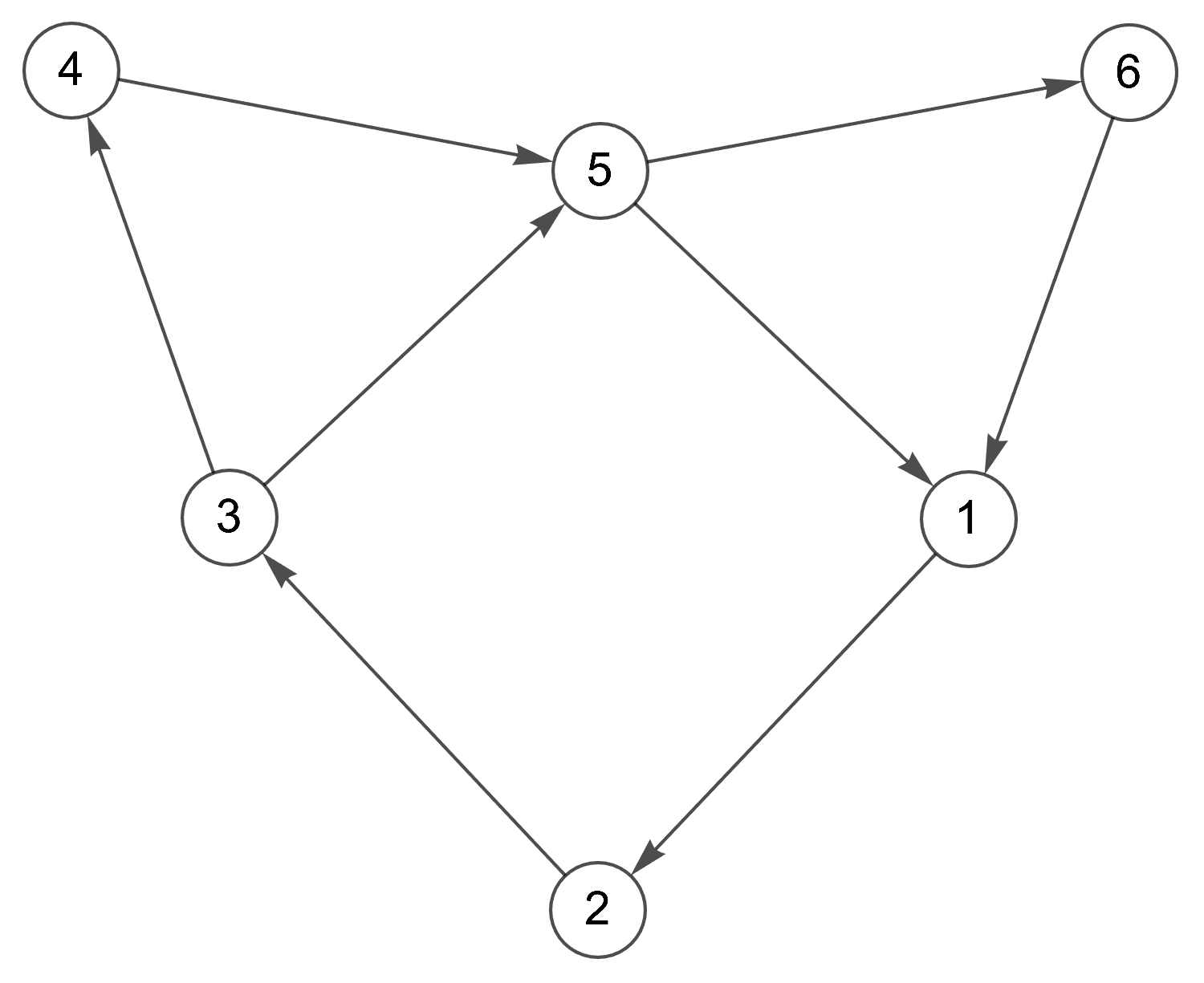}\\
\begin{minipage}{0.45\textwidth}\small\begin{align*}
    &H_1 = x_1+x_2+x_3+x_4+x_5+x_6\\
    &\begin{aligned}
        &H_2 = x_1 x_3+x_5 x_3+x_6 x_3+x_1 x_4+\\
        &\hspace*{4em}x_2 x_4+x_2 x_5+x_2 x_6+x_4 x_6\\
    \end{aligned}\\
    &H_3 = x_2 x_4 x_6\\
    &H_4 = x_1 x_2 x_3 x_4 x_5
\end{align*}\end{minipage} & 
\begin{minipage}{0.45\textwidth}\small\begin{align*}
    &H_1 = x_1+x_2+x_3+x_4+x_5+x_6\\
    &\begin{aligned}
    &H_2 = x_1 x_3+x_6 x_3+x_1 x_4+x_2 x_4+\\
    &\hspace*{4em}x_2 x_5+x_2 x_6+x_4 x_6
    \end{aligned}\\
    &H_3 = x_2 x_4 x_6\\
    &H_4 = x_1 x_2 x_3 x_5
\end{align*}\end{minipage} 
\end{tblr}
\caption{(Left) A skip-graph with one skipped vertex. (Right) A skip-graph with two skipped vertices. Note, both skipped vertices have the same parity.}
\label{fig:SkipVertex}
\end{figure}

\subsection{Cloned Family}
Recall that a graph with a cloned vertex has (at least) two vertices that share the same neighbourhood. The graph shown in \cref{fig:image3} is the simplest possible example of a graph that is integrable as a result of cloning. Notice that Vertex 4 is a clone of Vertex 3. 
\begin{figure}[htp!]
    \centering
    \includegraphics[width=0.45\textwidth]{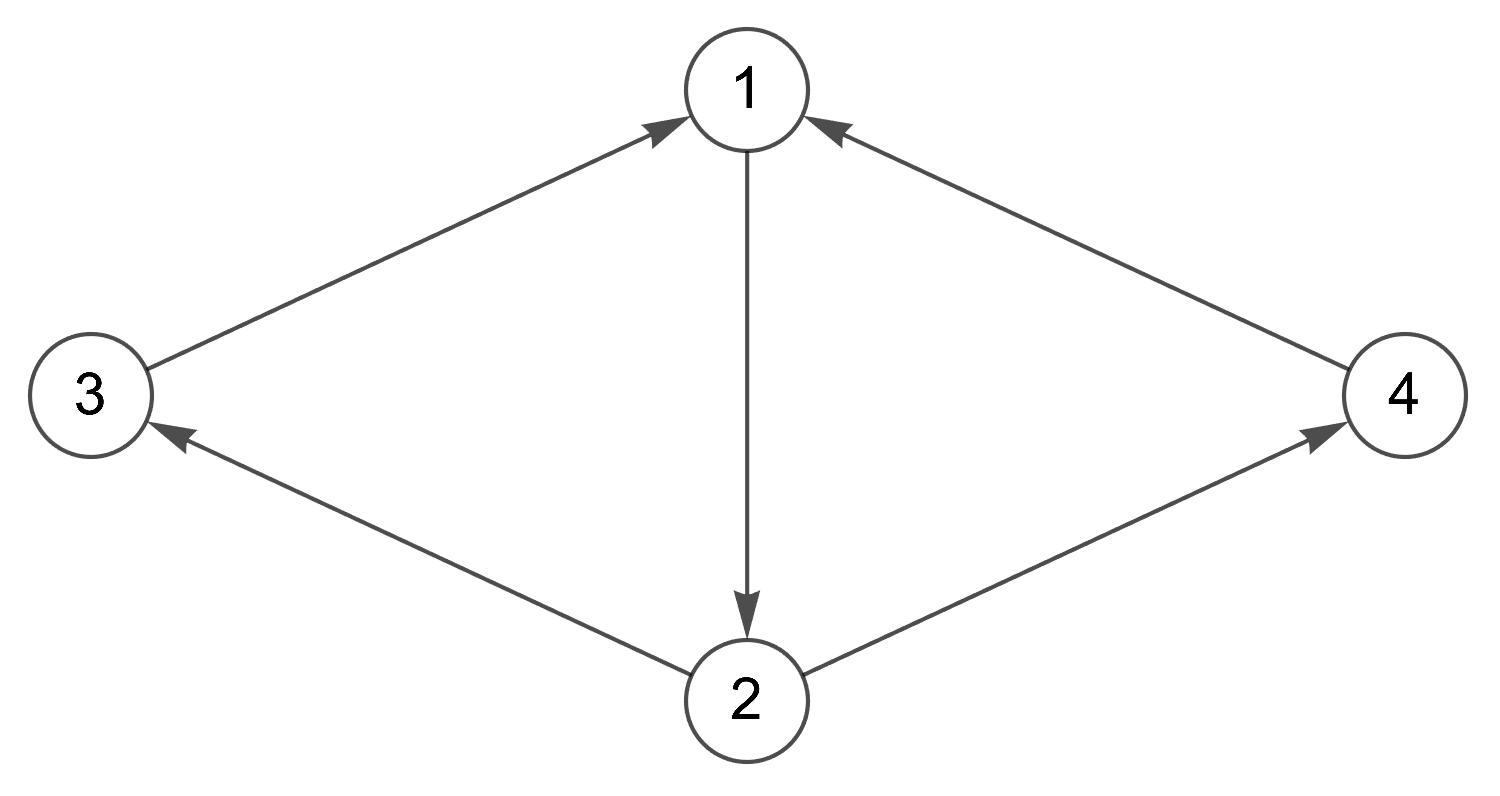}\\
    \begin{align*}
        &H_1 = x_1 + x_2 + x_3 + x_4\\
        &H_2 = x_1x_2x_3\\
        &H_3 = x_1x_2x_4
    \end{align*}
    \caption{The simplest example of cloning to produce an integrable graph. Vertex 3 has the same neighbourhood as Vertex 4. Decloning produces the integrable directed three-cycle.}
    \label{fig:image3}
\end{figure}
Removing Vertex 4, the operation referred to as decloning \cite{EKV22}, reduces the graph to the directed three cycle, which is integrable. It follows from \cref{thm:Declone} that the graph shown in \cref{fig:image3} is integrable. Evripidou  et al. provide a mechanism for constructing conserved quantities using Lax pairs \cite{EKV22}. There are 32 graphs that are determined to be integrable through  decloning. Interestingly, this is the largest subfamily of integrable graphs with six or fewer vertices. Notice this graph admits Itoh-style conserved quantities. 

\subsection{Embedding Family}
There are 12 integrable graphs that are purely derived from embeddings: one graph with five vertices, and eleven with six vertices. We illustrate a directed four-cycle embedded into a directed three-cycle in \cref{fig:image4}, for comparison to the example from \cref{fig:GraphEmbedding}.
\begin{figure}[htp!]
    \centering
    \includegraphics[width=0.45\linewidth]{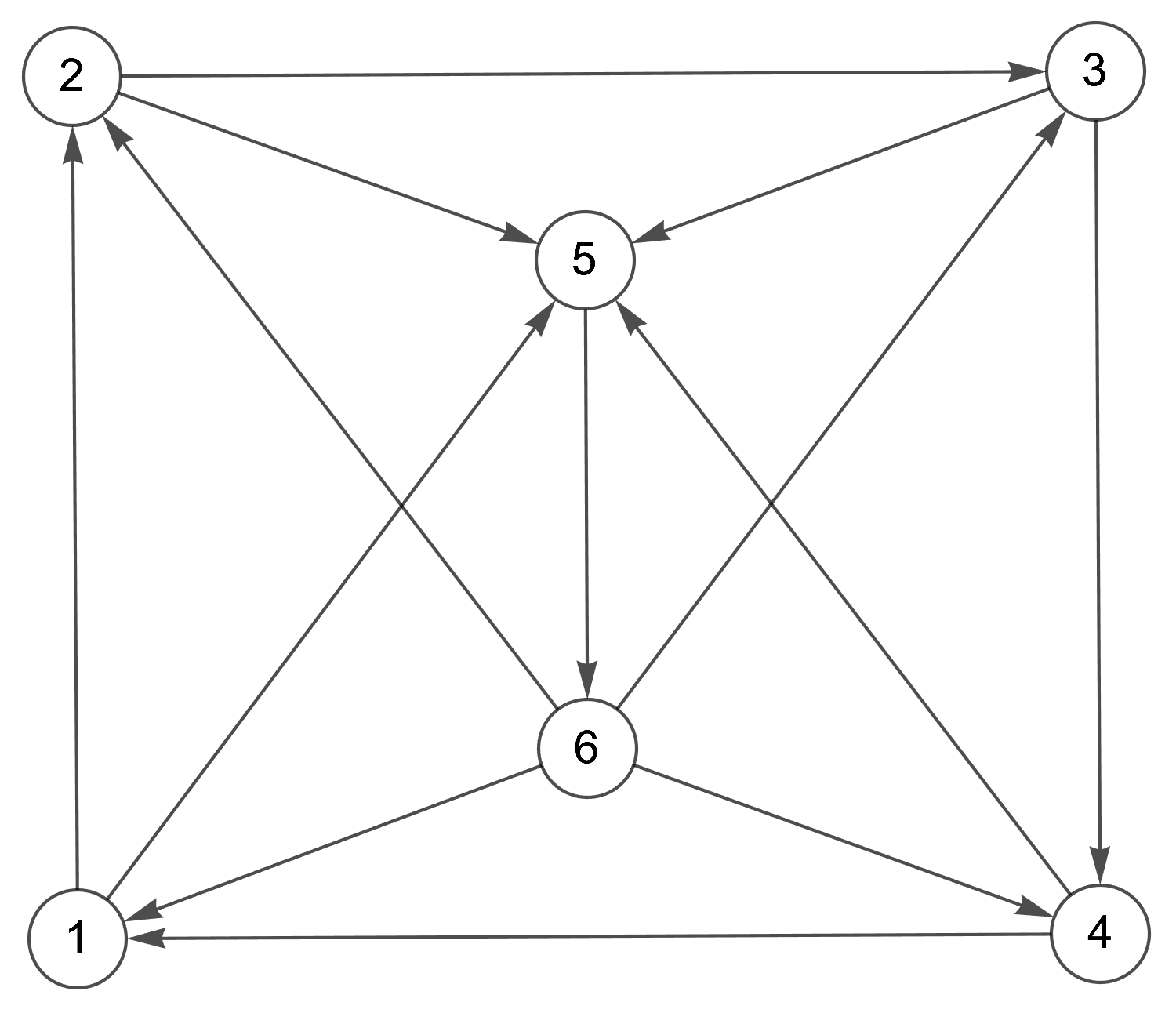}
    \begin{align*}
        &H_1 = x_1 + x_2 + x_3 + x_4 + x_5 + x_6\\
        &H_2 = (x_1 + x_2 + x_3 + x_4)x_5x_6\\
        &H_3 = x_1x_3x_5^2x_6^2\\
        &H_4 = x_2x_4x_5^2x_6^2
    \end{align*}
    \caption{The directed four-cycle embedded into a directed three-cycle}
    \label{fig:image4}
\end{figure}
Interestingly, there is overlap between the embedding family and the cloned family, as illustrated below. In our taxonomy, we count graphs in both the cloned and embedding families as embeddings first. We provide an example using the directed three-cycle embedded in the directed three-cycle, which was shown to be integrable by Paik and Griffin \cite{PG23} and was used by Allesina and Levine to model simple ecological systems \cite{AL11}. Therefore, we refer to this graph as the AL graph. To construct the graph shown in \cref{fig:image5}, one can either clone Vertex 5 to create Vertex 6 from the AL graph (shown in red) or embed a directed three-cycle into Vertex 1 of the graph shown in \cref{fig:image3}.
\begin{figure}[htp!]
    \centering
    \includegraphics[width=0.35\textwidth]{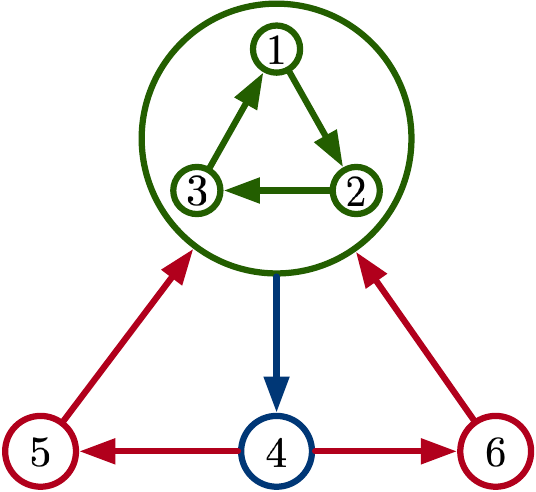}\\
    \begin{align*}
        &H_1 = x_1 + x_2 + x_3 + x_4 + x_5 + x_6\\
        &H_2 = \left(x_1+x_2+x_3\right) x_4 \left(x_5+x_6\right)\\
        &H_3 = x_1 x_2 x_3 x_4^3 x_5^3\\
        &H_4 = x_1 x_2 x_3 x_4^3 x_6^3\\
    \end{align*}
    \caption{Example of a graph that is both an embedding and a cloning. This graph is a cloned version of the AL graph, with Vertices 5 and 6 cloned. It's also an embedding of the 3 cycle into Vertex 1 of the graph shown in  \cref{fig:image3}.}
    \label{fig:image5}
\end{figure}
Notice this also nicely illustrates \cref{thm:Embed} in which we have a graph in the Bogoyavlenskij  family embedded into a graph in a graph that has Itoh-style conserved quantities but is not in the Bogoyavlenskij  family.

\bibliographystyle{iopart-num}
\bibliography{Integrability}

\providecommand{\newblock}{}
\begin{thebibliography}{10}
\expandafter\ifx\csname url\endcsname\relax
  \def\url#1{{\tt #1}}\fi
\expandafter\ifx\csname urlprefix\endcsname\relax\def\urlprefix{URL }\fi
\providecommand{\eprint}[2][]{\url{#2}}

\bibitem{A13}
Arnol'd V~I 2013 {\em Mathematical methods of classical mechanics\/} vol~60
  (Springer Science \& Business Media)

\bibitem{W64}
Wigner E~P 1964 {\em Proceedings of the National Academy of Sciences\/} {\bf
  51} 956--965

\bibitem{B06}
Butterfield J 2006 On symmetry and conserved quantities in classical mechanics
  {\em Physical theory and its interpretation: Essays in honor of Jeffrey
  bub\/} (Springer) pp 43--100

\bibitem{L-GPV12}
Laurent-Gengoux C, Pichereau A and Vanhaecke P 2012 {\em Poisson structures\/}
  vol 347 (Springer Science \& Business Media)

\bibitem{S76}
Smale S 1976 {\em Journal of Mathematical Biology\/} {\bf 3} 5--7

\bibitem{I87}
Itoh Y 1987 {\em Progress of theoretical physics\/} {\bf 78} 507--510

\bibitem{I08}
Itoh Y 2008 {\em Journal of Physics A: Mathematical and Theoretical\/} {\bf 42}
  025201

\bibitem{BIY08}
Bogoyavlenskij O, Itoh Y and Yukawa T 2008 {\em Journal of mathematical
  physics\/} {\bf 49}

\bibitem{EKV21}
Evripidou C, Kassotakis P and Vanhaecke P 2021 {\em Mathematical Physics,
  Analysis and Geometry\/} {\bf 24} 1--28

\bibitem{EKV22}
Evripidou C, Kassotakis P and Vanhaecke P 2022 {\em Journal of Physics A:
  Mathematical and Theoretical\/} {\bf 55} 325201

\bibitem{PG23}
Paik J and Griffin C 2023 {\em Physical Review E\/} {\bf 107} L052202

\bibitem{HS98}
Hofbauer J and Sigmund K 1998 {\em {Evolutionary Games and Population
  Dynamics}\/} ({Cambridge University Press})

\bibitem{HS03}
Hofbauer J and Sigmund K 2003 {\em Bulletin of the American Mathematical
  Society\/} {\bf 40} 479--519

\bibitem{W97}
Weibull J~W 1997 {\em Evolutionary Game Theory\/} (MIT Press)

\bibitem{T15}
Tanimoto J 2015 {\em Fundamentals of evolutionary game theory and its
  applications\/} (Springer)

\bibitem{T19}
Tanimoto J 2019 {\em Evolutionary Games With Sociophysics\/} (Springer)

\bibitem{AL84}
Akin E and Losert V 1984 {\em Journal of mathematical biology\/} {\bf 20}
  231--258

\bibitem{AL11}
Allesina S and Levine J~M 2011 {\em Proceedings of the National Academy of
  Sciences\/} {\bf 108} 5638--5642

\bibitem{GBMA17}
Grilli J, Barab{\'a}s G, Michalska-Smith M~J and Allesina S 2017 {\em Nature\/}
  {\bf 548} 210--213

\bibitem{MCLM23}
Miller Z~R, Clenet M, Libera K~D, Massol F and Allesina S 2023 {\em bioRxiv\/}
  (\textit{Preprint}
  \eprint{https://www.biorxiv.org/content/early/2023/07/13/2023.03.23.533867.full.pdf})
  \urlprefix\url{https://www.biorxiv.org/content/early/2023/07/13/2023.03.23.533867}

\bibitem{VS93}
Veselov A~P and Shabat A~B 1993 {\em Funktsional'nyi Analiz i ego
  Prilozheniya\/} {\bf 27} 1--21

\bibitem{M74}
Moser J 2005 {\em Dynamical Systems, Theory and Applications: Battelle Seattle
  1974 Rencontres\/}  467--497

\bibitem{KM75}
Kac M and van Moerbeke P 1975 {\em Advances in Mathematics\/} {\bf 16} 160--169

\bibitem{B88}
Bogoyavlensky O~I 1988 {\em Physics Letters A\/} {\bf 134} 34--38

\bibitem{CDPV15}
Charalambides S~A, Damianou P~A and Evripidou C~A 2015 Generalized
  lotka—volterra systems connected with simple lie algebras {\em Journal of
  Physics: Conference Series\/} vol 621 (IOP Publishing) p 012004

\bibitem{DEKV17}
Damianou P~A, Evripidou C~A, Kassotakis P and Vanhaecke P 2017 {\em Journal of
  Mathematical Physics\/} {\bf 58}

\bibitem{PMMC22}
Parmelee C, Moore S, Morrison K and Curto C 2022 {\em PloS one\/} {\bf 17}
  e0264456

\bibitem{PALC+22}
Parmelee C, Alvarez J~L, Curto C and Morrison K 2022 {\em SIAM journal on
  applied dynamical systems\/} {\bf 21} 1597--1630

\bibitem{GY18}
Gross J~L, Yellen J and Anderson M 2018 {\em Graph theory and its
  applications\/} (Chapman and Hall/CRC)

\bibitem{G21}
Griffin C 2021 {\em Chaos, Solitons \& Fractals\/} {\bf 153} 111508

\bibitem{OST13}
Ovsienko V, Schwartz R~E and Tabachnikov S 2013 {\em Duke Mathematical
  Journal\/} {\bf 162} 2149 -- 2196
  \urlprefix\url{https://doi.org/10.1215/00127094-2348219}

\bibitem{MR13}
Marsden J~E and Ratiu T~S 2013 {\em Introduction to mechanics and symmetry: a
  basic exposition of classical mechanical systems\/} vol~17 (Springer Science
  \& Business Media)

\bibitem{H96}
Hofbauer J 1996 {\em Journal of mathematical biology\/} {\bf 34} 675--688

\bibitem{B91}
Bogoyavlenskii O~I 1991 {\em Russian Mathematical Surveys\/} {\bf 46} 1

\bibitem{C1849}
Cayley A 1849

\bibitem{B04}
Bourbaki N 2004 {\em Elements of Mathematics, 2. Linear and multilinear
  algebra\/} (Springer Berlin Heidelberg)

\bibitem{P08}
Perrin D {\em et~al.\/} 2008 {\em Algebraic geometry: an introduction\/}
  (Springer)

\bibitem{M24}
McKay B 2024 Digraphs
  \url{https://users.cecs.anu.edu.au/~bdm/data/digraphs.html}

\bibitem{G23a}
Griffin C~H 2023 {\em Applied Graph Theory: An Introduction with Graph
  Optimization and Algebraic Graph Theory\/} (World Scientific)

\bibitem{WSSV85}
Wolf A, Swift J~B, Swinney H~L and Vastano J~A 1985 {\em Physica D: nonlinear
  phenomena\/} {\bf 16} 285--317

\bibitem{S96}
Sandri M 1996 {\em Mathematica Journal\/} {\bf 6} 78--84

\bibitem{GSB22}
Griffin C, Semonsen J and Belmonte A 2022 {\em Physica A: Statistical Mechanics
  and its Applications\/} {\bf 597} 127281

\end{thebibliography}
\end{document}